\documentclass[twocolumn,10pt]{IEEEtran}
\usepackage{algorithm,algorithmic,amsbsy,amsmath,amssymb,epsfig,bbm,mathrsfs,fancyhdr,fancyvrb,subfigure,url,color}

\usepackage{bbm}
\usepackage{cite}

\newcommand{\beq}{\begin{equation}}
\newcommand{\eeq}{\end{equation}}
\newcommand{\beqn}{\begin{eqnarray}}
\newcommand{\eeqn}{\end{eqnarray}}

\newtheorem{theorem}{\textbf{\text{Theorem}}}
\newtheorem{lemma}{\textbf{\text{Lemma}}}

\newenvironment{proof}[1][Proof]{\begin{trivlist}
\item[\hskip \labelsep {\bfseries #1}]}{\end{trivlist}}

\newcommand{\qed}{\nobreak \ifvmode \relax \else
      \ifdim\lastskip<1.5em \hskip-\lastskip
      \hskip1.5em plus0em minus0.5em \fi \nobreak
      \vrule height0.75em width0.5em depth0.25em\fi}

\begin{document}

%\title{A Tractable Uplink Modeling Paradigm for Multi-tier Cellular Wireless Networks}

\title{On Stochastic Geometry Modeling of Cellular Uplink Transmission with Truncated Channel Inversion Power Control}

\author{Hesham ElSawy and Ekram Hossain 
\thanks{The authors are with the Department of Electrical and Computer Engineering, University of Manitoba, Winnipeg, Canada (emails:  umelsawy@cc.umanitoba.ca, Ekram.Hossain@umanitoba.ca).  
This work was supported in part by an NSERC Strategic  Grant (STPGP 430285), and in part by a scholarship from TR{\em Tech}, Winnipeg, Manitoba, Canada.
%This work was supported in part by a scholarship from TR{\em Labs}, Winnipeg, and in part by the NSERC, Canada, Discovery Grant 249500-2009 awarded to \textbf{E. Hossain} (corresponding author).
} 

\vspace{-2mm}
}

\maketitle

\begin{abstract}
%There is a developing interest in stochastic geometry analysis for cellular networks due to its tractability and accuracy. However, there has not been a unified analytical model for uplink transmissions. In this paper, we develop a tractable uplink modeling paradigm for outage probability and rate in muli-tier cellular networks.
%Despite the developing interest in stochastic geometry analysis for cellular networks, there has not been a unified analytical framework yet for uplink transmissions. 
Using stochastic geometry, we develop a tractable uplink modeling paradigm for outage probability and spectral efficiency in both single and multi-tier cellular wireless networks. The analysis accounts for per user equipment (UE) power control as well as the maximum power limitations for UEs. More specifically, for interference mitigation and robust uplink communication, each UE is required to control its transmit power such that the average received signal power at its serving base station (BS) is equal to a certain threshold $\rho_o$. Due to the limited transmit power, the UEs employ a truncated channel inversion power control policy with a cutoff threshold of $\rho_o$. We show that there exists a transfer point in the uplink system performance that depends on the tuple: BS intensity ($\lambda$), maximum transmit power of UEs ($P_u$), and $\rho_o$. That is, when $P_u$ is a tight operational constraint with respect to [w.r.t.] $\lambda$ and $\rho_o$, the uplink outage probability and spectral efficiency highly depend on the values of $\lambda$ and $\rho_o$. In this case, there exists an optimal cutoff threshold $\rho^*_o$, which depends on the system parameters, that minimizes the outage probability. On the other hand, when $P_u$ is not a binding operational constraint w.r.t. $\lambda$ and $\rho_o$, the uplink outage probability and spectral efficiency become independent of $\lambda$ and $\rho_o$. We obtain approximate yet accurate simple expressions for outage probability and spectral efficiency which reduce to closed-forms in some special cases. 

%and simple forms that include a single numerical integral for the . %The analysis shows that the uplink performance improves with increasing the intensity of BS, however, in relatively dense deployment of BSs, both the outage probability and spectral efficiency become independent from the BSs density. That is, increasing the density of BSs will not affect the either outage probability or the spectral efficiency.    


\end{abstract}

\section{Introduction} \label{intro}

Due to the variations of the capacity demand across the service areas (e.g., downtowns, residential and suburbans areas, etc.) as well as the deployment of multi-tier cellular networks (i.e., cellular networks consisting of different types of base stations [BSs] such as macro, micro, pico, and femto BSs), the practical cellular infrastructure deviates from the well-known regular grid topology. Instead, it is better modeled by a randomized topology with uncertainties in the locations of the BSs~\cite{trac, trac2}. Therefore,  modeling and analysis of  cellular networks based on stochastic geometry has recently received much attention due to its tractability and accuracy \cite{our-survey}. Stochastic geometry is a very powerful tool to deal with networks with random topologies. Stochastic geometry abstracts the cellular network topology to a convenient point process to facilitate its modeling and analysis. The Poisson point process (PPP) is an appealing abstraction for the cellular network topology due to its simplicity and tractability~\cite{our-survey, martin-book}. Furthermore, the PPP has been proven to provide pessimistic bounds on the performance metrics (e.g., the outage probability and the mean rate)  that are as tight as the optimistic bounds provided by the idealized grid-based model for actual cellular networks \cite{trac, trac2}. Further validations of stochastic geometry-based modeling of cellular networks via PPP can be found in \cite{pp-cellular, valid, wyner-acc}. 
%Furthermore, the PPP has been proven to provide lower bounds on the coverage probability (i.e., the complement of the outage probability) and the mean rate for actual cellular network that are as tight as the upper bound provided by the idealized grid-based model \cite{trac, trac2}. Further validations of stochastic geometry-based modeling of cellular networks via PPP can be found in \cite{pp-cellular, valid, wyner-acc}. 

There have been significant developments in the stochastic geometry modeling of cellular wireless networks. However, most of the available literature focus on modeling and analysis of downlink transmissions due to its relative simplicity. In a multi-tier cellular network, if the users associate to the BSs in the downlink based on their average received signal strength, the average useful signal power received at each user equipment (UE) from its corresponding BS will be strictly greater than the average interference power from any individual interfering BS. Therefore, power control is not crucial\footnote{Indeed power control is important in downlink to enhance the network performance (see \cite{Pcon}) %e.g., in terms of power efficiency, outage, and rate) 
but not crucial for the basic network operation due to the inherent interference protection introduced by the association criterion.} for the network operation, and hence, power control is ignored in most of the stochastic geometry models on downlink cellular networks and it is generally assumed that all the BSs in the same tier transmit with equal power \cite{trac, trac2, our-survey, sayandev, our-tmc, tony, %son,
 our-cog1, cao, MIMO3}.
%In a single-tier cellular network, in the downlink, if the users associate to the BSs  based on their average received signal strength, each user equipment (UE) will be closer to its serving BS than any other interfering BS. In such a scenario, the useful average signal power received at a generic UE from its serving BS is strictly greater than any individual interfering signal power received from an interfering BS. Therefore, power control is not crucial\footnote{Indeed power control is important in downlink to enhance the network performance (see \cite{Pcon}) %e.g., in terms of power efficiency, outage, and rate) but not crucial for the basic network operation due to the inherent interference protection introduced by the association criterion.} for the network operation, and hence, power control is ignored in most of the stochastic geometry models on downlink cellular networks and it is generally assumed that all the BSs transmit with equal power \cite{trac, trac2, our-survey, sayandev, our-tmc, tony, son, our-cog1, cao, MIMO3}. 
In contrast, uplink analysis is quite more involved due to the following reasons:
\begin{itemize}
\item \textbf{Per user power control:} in the uplink, due to the random cell sizes, an interfering UE in a neighboring cell can be much closer to a BS than its tagged UE (cf. Fig.~\ref{per_chan}). Therefore, power control per UE is crucial for basic uplink operation in order to mitigate the inter-cell interference. As will be shown later, per UE power control introduces a new source of randomness to the uplink system model which makes the uplink analysis more involved. 

\item \textbf{Correlation among interferers:} orthogonal channel assignment per BS ensures no channel reuse in the same Voronoi\footnote{A Voronoi tessellation is the planar graph constructed by perpendicular lines bisecting the distances between the points of a point process. The Voronoi cells constructed from a point process representing the  BS locations correspond to the coverage regions of the BSs.} cell. That is, given that a UE is transmitting in the uplink on a certain channel, this channel cannot be reused within the coverage of its serving BS, and hence, the locations of the UEs using the same uplink channel are correlated. 
\end{itemize}

\begin{figure}[t]
	\begin{center}
		\includegraphics[trim = 3.5cm 7.5cm 3.5cm 7.3cm , clip, width=2.5 in]{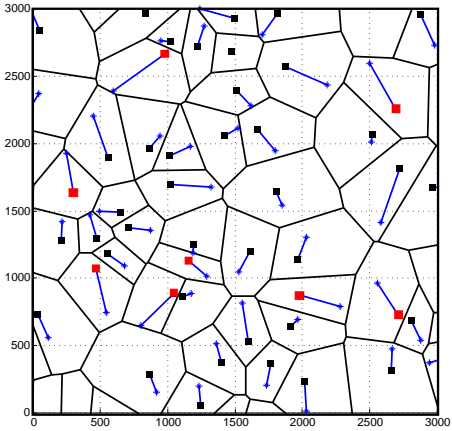}
	  \end{center}
	\caption{Uplink network model representing the UEs served per channel, where the BSs are represented by squares, the UEs are represented by stars, and the lines denote the UEs' association. BSs having an interfering UE closer than their tagged UE are highlighted in red.}
	\label{per_chan}
\end{figure}

Compared to the downlink, only few efforts were invested to understand and model the uplink transmissions in cellular networks. Uplink modeling has been done in an ad hoc manner where different works in the literature made different assumptions based on the problem in hand. For instance, \cite{and-up2} derived the uplink network capacity region for a two-tier cellular network consisting of macro BSs (MBSs) modeled via the hexagonal grid, femto access points (FAPs) modeled via PPP, and UEs modeled via an independent PPP. Due to the small coverage radius of the FAPs, in \cite{and-up2}, the interference seen from all uplink UEs associated with the same FAP was approximated by an isotropic point source of interference with the worst-case sum transmit powers of the FAP UEs (i.e., power control was ignored). \cite{spec-share} investigated the spectrum sharing between a cellular network uplink and a mobile ad hoc network. In \cite{spec-share}, the service areas of all BSs were approximated by circles and power control was ignored. Note that the assumption of circular coverage areas for the BSs eliminates the aforementioned interference problem in the uplink where an interfering UE can be much closer to a BS than the tagged UE which communicates to that BS. Therefore,  under the assumption of circular coverage area, power control can be ignored. 

In \cite{and-uplink}, the authors modeled uplink transmission with fractional channel inversion power control in a single-tier cellular network. However, for analytical tractability and to avoid the singularity at the origin imposed by the unbounded path-loss model\cite{our-survey}, the authors in \cite{and-uplink} approximated the entire network model and assumed that the Voronoi cells are divided with respect to (w.r.t.) the users (rather than w.r.t.  the BSs) and each user has her own serving BS randomly located in her Voronoi cell. In \cite{and-d2d}, the authors modeled uplink UEs with channel inversion power control. However, the authors assumed that the tagged UE is uniformly distributed in the tagged BS's coverage which is approximated by a circle having the radius $\frac{1}{\sqrt{\pi \lambda}}$, where $\lambda$ is the BS intensity\footnote{More discussions on the work presented in \cite{and-uplink, and-d2d} will be provided in Section \ref{resss}.}. Each of the available stochastic geometry models \cite{and-up2, spec-share, and-uplink, and-d2d} approximates the system model in a different way to simplify the analysis and none of them accounts for the maximum transmit power of the UEs. 

In this paper, we propose a framework for modeling and analysis of uplink transmission in single and multi-tier cellular networks. We resort to some approximations in order to maintain the tractability of analysis. In particular, the proposed framework, as will be discussed later in details,  partially ignores the correlations among the interfering UEs. That is, it captures the correlation between  the interfering sources and the test network elements (i.e., the tagged UE and the corresponding BS) and ignores the mutual correlations among the interfering sources. The accuracy of the proposed analysis is validated via simulations. To the best of the authors' knowledge, when compared to the existing literature on stochastic geometry analysis of uplink cellular systems, this paper uses the least amount of approximations and obtains the simplest expressions for outage probability and spectral efficiency of uplink transmission. The obtained simple expressions characterize the uplink system performance in both single and multi-tier cellular networks and reveal important design insights. We believe that the presented framework, results, and design insights fill in the gap between the well-understood downlink performance and the lagging uplink analysis. The contributions of the paper can be summarized in the following points:

\begin{itemize}
\item The paper presents a tractable framework for uplink modeling 
%without approximating the system model 
and analysis in a Poisson cellular network. The model is general and extends to multi-tier cellular networks. The model accounts for limited transmit power of the UEs, per UE power control, and cutoff threshold for the power control. 

\item Approximate yet accurate simple closed-form equations are derived for the outage probability and simple forms with only one numerical integral are derived for the spectral efficiency.

\item The paper discusses the tradeoffs introduced by the cutoff threshold and the maximum transmit power and shows that there exists an optimal cutoff threshold for power control that minimizes the outage probability and  transmit power of the UEs.

\item The paper characterizes the uplink performance and shows the commonalities and differences between the downlink and uplink performances in cellular networks. In particular, we show the existence of a transfer point in the uplink system performance which depends on the tuple ($\lambda$, $P_u$, $\rho_o$), where $\lambda$ is the BS intensity, $P_u$ is the maximum transmit power of a UE, and $\rho_o$ is the average received power required at the serving BS. That is, when the relative values of $\lambda$, $P_u$, and $\rho_o$ lead to a binding maximum transmit power constraint for the uplink operation, the uplink operation is quite different from the downlink operation and depends on both the cutoff threshold and the BS intensity.
In contrast, when the relative values of $\lambda$, $P_u$, and $\rho_o$ lead to a non-binding maximum transmit power constraint for the uplink operation (i.e., transmissions are not constrained by the maximum transmit power), the uplink operation becomes analogous to the downlink operation (i.e., becomes independent of the BS intensity and cutoff threshold).

\end{itemize}

The rest of the paper is organized as follows. The system model, assumptions, and methodology of analysis are presented in Section II. In Section III, the baseline uplink modeling paradigm for single-tier cellular networks is presented. Section IV generalizes the developed uplink paradigm for multi-tier cellular networks. Numerical and simulation results are presented in Section V and the paper is concluded in Section VI.

\section{System Model and Assumptions}

\subsection{Network Model}

We consider an independent $K$-tier Poisson cellular network. That is, the BSs of each tier are spatially distributed in $\mathbb{R}^2$ according to an independent homogenous PPP $\mathbf{\Psi}_k=\left\{m^{(k)}_i; i=1,2,3,...\right\}$, $k \in \left\{1,2,3,...,K\right\}$ with intensity $\lambda_k$, where $m^{(k)}_i \in \mathbf{\Psi}_k$ is the location of the $i^{th}$ BS in the $k^{th}$ tier. The UEs are spatially distributed in $\mathbb{R}^2$ according to an independent PPP $\mathbf{\Phi}=\left\{u_i; i=1,2,3,...\right\}$ with intensity $\mathcal{U}$.\footnote{With a slight abuse of notation we will use $m^{(k)}_i$ to denote both the location of the $i^{th}$ BS in the $k^{th}$ tier and the BS itself, and the same for $u_i$.} It is assumed that the intensity of the UEs  is high enough such that each BS will have at least one user served per channel and that UEs have data to transmit in the uplink (i.e., saturation conditions are assumed). The BSs in each tier $k$ have equal receiver sensitivity $\rho^{(k)}_{min}$, however, two BSs from different tiers do not necessarily have the same receiver sensitivity. For successful uplink communication, it is required that the received signal power at the BS is greater than the receiver sensitivity. Therefore, each of the UEs associated to tier $k$ controls its transmit power such that the average signal received at the corresponding serving BS is equal to the threshold $\rho^{(k)}_o$, where $\rho^{(k)}_o > \rho^{(k)}_{min}$. 

It is assumed that all the UEs have the same maximum transmit power of $P_u$. Due to the maximum transmit power constraint $P_u$ for uplink communications, the UEs  use truncated channel inversion power control where the transmitters compensate for the path-loss to keep the average received signal power equal to the threshold $\rho^{(k)}_o$ \cite[chapter 4]{andrea-book}. That is, an uplink connection is established between a UE and its serving BS if and only if the transmit power required for the path-loss inversion is less than $P_u$. Otherwise, the UE does not transmit and goes into an outage (denoted hereafter as {\em truncation outage}) due to the insufficient transmit power. 

It is worth mentioning that the truncated channel inversion power control mechanism is a realistic power control scheme for code-division multiple access (CDMA) networks to eliminate the near-far effect. Moreover, for orthogonal frequency-division multiple access (OFDMA) networks, it has been shown in \cite{hina} that if the edge users (i.e., users who do not have sufficient power for their channel inversions) are allowed to transmit with their maximum power, the interference in the system increases significantly. Consequently, the network performance deteriorates without much improvement in the cell edge user performance. Hence, in this paper we consider the truncated channel inversion power control with the cutoff threshold $\rho^{(k)}_o$, where the cutoff threshold $\rho^{(k)}_o$ is a network design parameter that highly impacts the system behavior. As will be shown later, the relative values of the BS intensity, $P_u$, and $\rho^{(k)}_o$ control the tradeoff among transmit power, signal-to-interference-plus-noise-ratio ({\rm SINR}), and truncation outage (i.e., outage due to insufficient transmit power). 

Fig.~\ref{network} shows the network model for different values of $\rho_o$ for a single-tier cellular network. As shown in Fig.~\ref{net1}, if the value of $\rho_o$ is relatively high (i.e., relative to $\lambda$ and $P_u$), not all the UEs can compensate for the path-loss  and the cell edge UEs suffer from {\em truncation outage}. In contrast, Fig.~\ref{net2} shows that if the value of $\rho_o$ is relatively low (or equivalently, the BSs are dense enough and/or the maximum transmit power is high enough), all of the UEs can compensate for the path-loss and none of the UEs suffers from truncation outage. Hereafter, we will denote the scenario in Fig.~\ref{net1} as the uplink operation under binding maximum transmit power constraint  and the scenario in Fig.~\ref{net2} as the uplink operation under non-binding maximum transmit power constraint. Universal frequency reuse is used within each tier and across different tiers. Within a network tier, there is no intra-cell interference. That is, in a network tier, each BS assigns a unique channel to each of its associated UEs. 

\begin{figure}[t]
	\begin{center}
		\subfigure[]{\label{net1}\includegraphics[width=2.5 in]{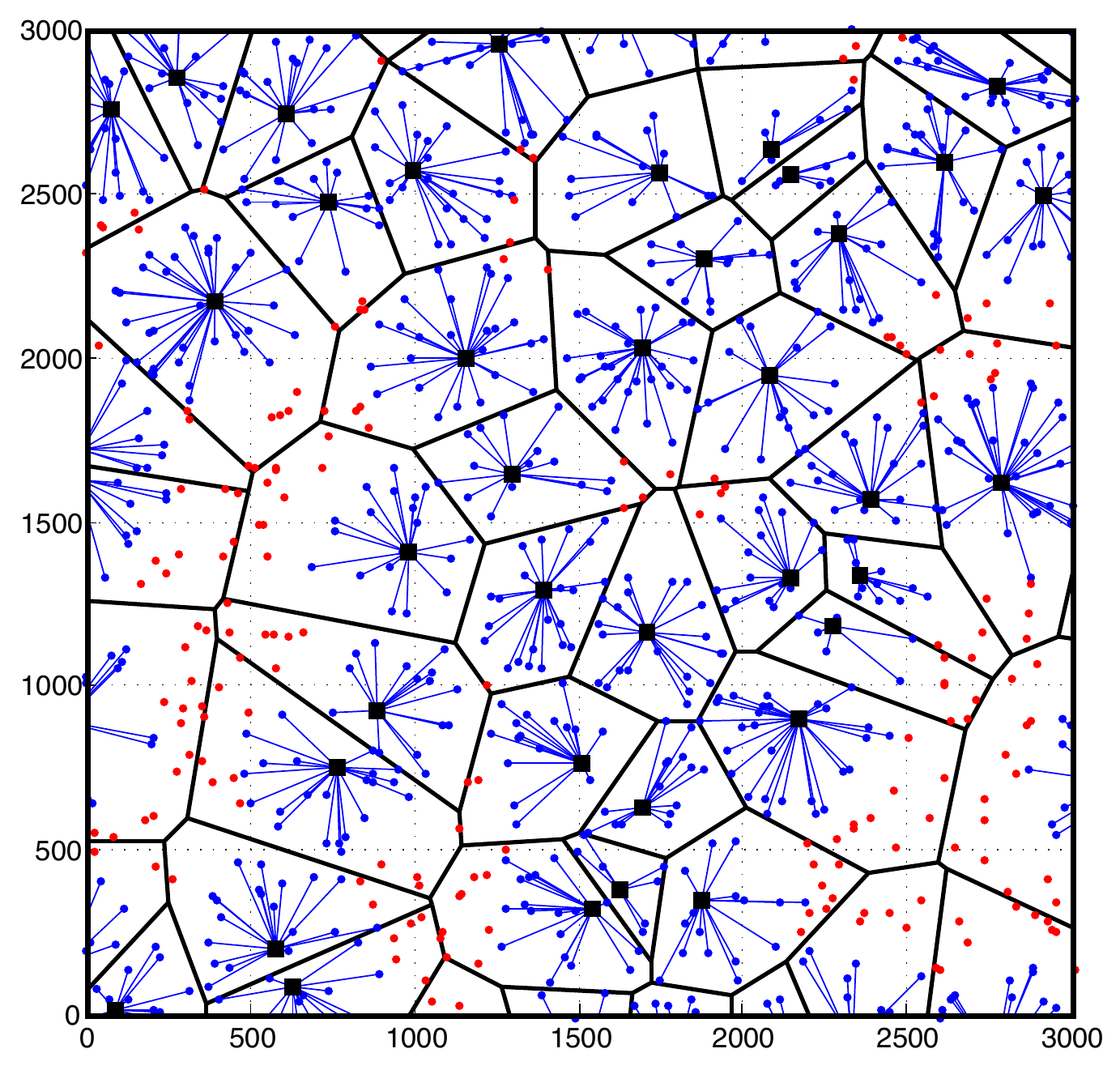}}
	  \subfigure[]{\label{net2}\includegraphics[width=2.5 in]{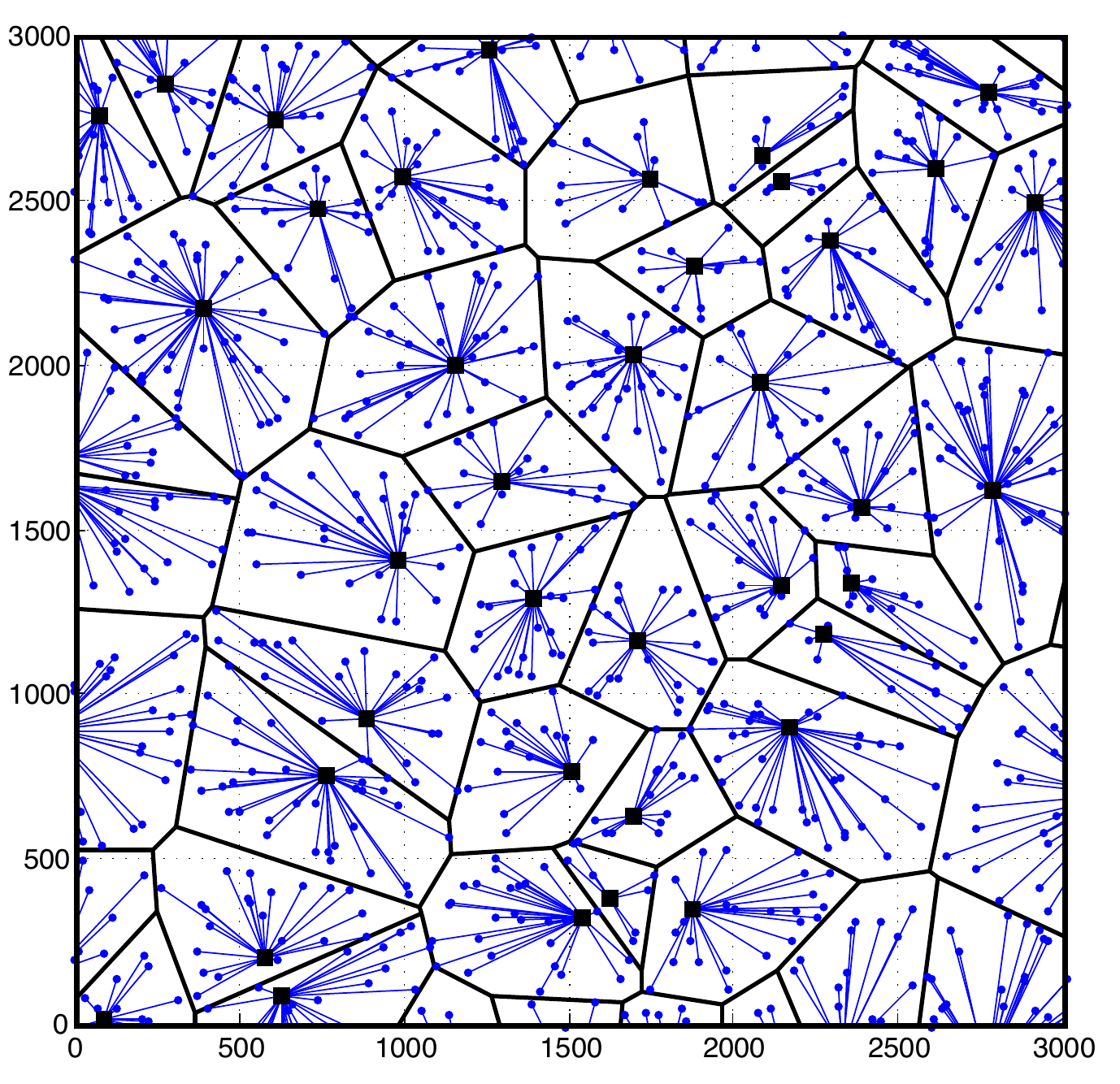}}
	\end{center}
	\caption{ Network model showing the served users on all channels in $[0,3000]^2$ m$^2$; the black squares represent the BSs, the blue dots represent the active UEs, the blue lines denote the UEs' associations, and the red dots are the inactive UEs due to the insufficient transmit power. The simulation parameters are $\lambda = 5$ BS/km$^2$, $\mathcal{U}=100$ UE/km$^2$, and (a) $\rho_o = -70$ dBm, (b) $\rho_o = -90$ dBm.}
	\label{network}
\end{figure}

%Looking to the system model in , in the Uplink we can see cases where the interfering UE can be much closer to a BS than its intended UE. Therefore, power control is curial in the uplink to mitigate the inter-cell interference.  

\subsection{Radio Channel Model} \label{Rmodel}

A general power-law path-loss model is considered in which the signal power decays at the rate $r^{-\eta}$ with the propagation distance $r$, where $\eta$ is the path-loss exponent. For simplicity, it is assumed that all BSs in all tiers share the same path-loss exponent $\eta$. The case where each network tier has its own path-loss exponent $\eta_k$ will be discussed in Section~\ref{multi}. The channel (power) gain between two generic locations $x,y \in \mathbb{R}^2$ is denoted by $h(x,y)$. All the channel gains are assumed to be independent of each other, independent of the spatial locations, symmetric and are identically distributed (i.i.d.). Therefore, for the brevity of exposition, hereafter, the spatial  indices $x,y$ are dropped. For analysis, only Rayleigh fading  is assumed\footnote{Techniques to relax the Rayleigh fading assumption to general fading channels can be found in \cite{our-survey}.} and the channel power gain $h$ is assumed to be exponentially distributed with unit mean. A signal-to-interference-plus-noise ratio ({\rm SINR}) model is considered where a message can be successfully decoded at the tagged BS in tier $k$ if and only if the {\rm SINR} of the useful signal is greater than a certain threshold $\theta_k$. %and the useful signal power is greater than the receiver sensitivity $\rho^{(k)}_{min}$. 
If the {\rm SINR} at the tagged BS does not exceed the threshold $\theta_k$, %or the useful signal power is below the receiver sensitivity $\rho^{(k)}_{min}
 the link experiences an outage (hereafter referred to as {\em {\rm SINR} outage}). 

\subsection{Criterion for Uplink Association}

Without loss of generality, all BSs in all tiers are assumed to have an open access policy\footnote{Closed access policy can be easily captured by thinning the PPP representing the complete set of BSs in the association analysis and simple modifications in the interference analysis~\cite{our-cog1}.}, and hence, all UEs can associate with all BSs. The UEs are assumed to associate to the BSs according to their  average link quality\footnote{Note that user association is not based on the instantaneous link quality to avoid the ping-pong effect in handovers that may arise due to the rapid fluctuations of the channel gains~\cite{req1}.}. 
%When all the network tiers share a common path-loss exponent, the association regions of the BSs will form a Voronoi tessellation. 
That is, a generic UE $u$ associates with its nearest BS from all tiers (i.e., $m^{(l)}_i$ if $\left\|u-m^{(l)}_i\right\| < \underset{m^{(k)}_j \in \left\{\mathbf{\Psi}_k\setminus m^{(l)}_i \right\}}\min\left\|u-m^{(k)}_j\right\|$, $\forall k$, where $\left\|.\right\|$ denotes the Euclidean norm). 

%It is worth mentioning that, in a multi-tier network, due to the heterogeneous transmit powers of the BSs in different network tiers, the downlink association regions of the BSs form a {\em weighted Voronoi} tessellation \cite{trac2, our-tmc}. In contrast, in the uplink, the homogenous transmit powers of the UEs result in association regions in the form of a {\em Voronoi} tessellation. 

%Therefore, the uplink association is different from the downlink association in multi-tier cellular networks \cite{7-way}. That is, the UE might not be associated to the same BS in the uplink and the downlink. Note that the LTE-A standard defines the coordinated multi-point (COMP) transmission to allow flexible and different uplink and downlink association \cite[chapter 13]{harri}.   

\subsection{Modeling Methodology}

The {\rm SINR} is a very important parameter  that affects many performance metrics such as outage, rate, delay, and energy efficiency. We characterize the {\rm SINR} by driving its cumulative distribution function ({\em cdf}) \cite{our-survey}. %To derive the {\em cdf} of {\rm SINR} on a given channel, it is required to characterize the activity of the UEs on that channel as well as their transmit powers. 
Due to the randomized network topology, the distances between the UEs and their serving BSs are random. Therefore, the transmit powers of the UEs are random (due to the truncated channel inversion power control).

We will first characterize the transmit power of a generic active UE (i.e., a generic user not in truncation outage) by deriving its probability density function ({\em pdf}) and moments. Then, we derive the {\em cdf} of {\rm SINR}. We first develop the modeling paradigm for single-tier cellular networks. Then, we show that the developed modeling paradigm for single-tier cellular networks can be naturally extended to multi-tier networks.  

\section{Uplink Modeling in a Single-tier Cellular Network}

In this section, we develop the baseline uplink modeling framework for a single-tier cellular network. For the sake of an organized presentation, we further divide this section into two subsections, namely, transmit power analysis and {\rm SINR} analysis. 

\subsection{Transmit Power Analysis} \label{tx}

Due to the random network topology and the use of truncated channel inversion power control, each UE will transmit with a different power to invert the path-loss towards its serving BS. In this section, we derive the distribution and the moments of the transmit power of a generic UE. Fig.~\ref{network} shows the uplink association for a single-tier cellular network for different values of the cutoff threshold $\rho_o$. As shown in Fig.~\ref{net1}, due to the truncated channel inversion power control, not all of the UEs can communicate in the uplink when the cutoff threshold is relatively high (i.e., relative to $P_u$ and $\lambda$). That is, the UEs located at a distance greater than $\left(\frac{P_u}{\rho_o}\right)^\frac{1}{\eta}$ from their nearest BS are unable to communicate in the uplink due the insufficient transmit power. Therefore, the complete set of UEs is divided into two non-overlapping subsets, namely, the subset of {\em active} UEs and the subset of {\em inactive} UEs. The inactive UEs do not transmit and experience outage due to insufficient transmit power. The distribution for the transmit power of a generic active UE is obtained from the following lemma.

\begin{lemma}
\label{lem1}
In a single-tier Poisson cellular network with truncated channel inversion power control with the cutoff threshold $\rho_o$, the {\em pdf} of the transmit power of a generic active UE in the uplink is given by
 \begin{align}
f_{P}(x) &=  \frac{2 \pi \lambda x^{\frac{2}{\eta}-1} e^{-\pi \lambda \left(\frac{x}{\rho_o}\right)^{\frac{2}{\eta}}}}{\eta  \rho_o^{\frac{2}{\eta}} \left(1-e^{-\pi \lambda \left(\frac{P_u}{\rho_o}\right)^{\frac{2}{\eta}}}\right)}, 0\leq x \leq P_u. \notag 
\end{align}
The moments of the transmit power can be obtained as 
\begin{equation}
\mathbb{E}\left[P^\alpha\right] =   \frac{\rho_o^\alpha \gamma\left(\frac{\alpha \eta}{2}+1, \pi \lambda \left(\frac{P_u}{\rho_o}\right)^\frac{2}{\eta}\right)}{(\pi \lambda)^{\frac{\alpha \eta}{2}}\left(1-e^{-\pi \lambda \left(\frac{P_u}{\rho_o}\right)^{\frac{2}{\eta}}}\right)}
\end{equation}  
where $\gamma(a,b)=\int_0^b t^{a-1} e^{-t} dt$ is the lower incomplete gamma function.
\end{lemma}

\begin{proof}
See \textbf{Appendix \ref{proof1}}.
\end{proof}

\vspace{0.2cm}
\textbf{Lemma~\ref{lem1}} shows that the smaller the cutoff threshold $\rho_o$, the lower is the  transmit power of the UEs. That is, when the maximum transmit power is unlimited (i.e., $\underset{P_u \rightarrow \infty}{\lim}\mathbb{E}\left[P^\alpha\right] =   \frac{\rho_o^\alpha \Gamma\left(\frac{\alpha \eta}{2}+1\right)}{(\pi \lambda)^{\frac{\alpha \eta}{2}}}$, where $\Gamma(.)$ is the gamma function), the expected transmit power linearly increases with $\rho_o$. This is because as $\rho_o$ increases the UEs are required to transmit at a higher power. \textbf{Lemma~\ref{lem1}} also shows that the average transmit power decreases with $\lambda$. That is, as the BS intensity increases, the distance between a generic UE and the corresponding serving  BS decreases, and hence, a lower transmit power is required to invert the path-loss. The truncation outage probability (i.e., the probability that a UE experiences outage due to the insufficient power) is given by
\begin{equation} 
\mathcal{O}_p = e^{-\pi \lambda \left(\frac{P_u}{\rho_o}\right)^{\frac{2}{\eta}}}
\label{op}
\end{equation} 
which is increasing in $\rho_o$. From \textbf{Lemma~\ref{lem1}} and equation (\ref{op}) it appears that the lower the cutoff threshold, the better is the network performance in terms of truncation outage probability and transmit power consumption. However, as will be discussed later, a low $\rho_o$ may highly deteriorate the {\rm SINR} outage and spectral efficiency, and hence, $\rho_o$ introduces a tradeoff for the network performance. Equation (\ref{op}) shows that the truncation outage probability exponentially decreases with increasing BS intensity.
%which is also non-decreasing in $\rho_o$. From \textbf{Lemma~\ref{lem1}} and equation (\ref{op}) it appears that the lower the cutoff threshold, the better is the network outage performance. However, as will be discussed later, a low $\rho_o$ may significantly deteriorate the {\rm SINR} performance and spectral efficiency, and hence, $\rho_o$ introduces a tradeoff among the different network performance parameters. Equation (\ref{op}) shows that the truncation outage probability exponentially decreases with the increase in BS density.

\subsection{{\rm SINR} Analysis} \label{outage}

In this section we derive the {\rm SINR} {\em outage} probability for active UEs (i.e., users not in truncation outage). Note  that the inactive UEs do not transmit  and are in truncation outage due to the insufficient transmit power. Without any loss in generality, the {\rm SINR} analysis is conducted on a tagged BS located at the origin. According to Slivnyak's theorem~\cite{martin-book}, conditioning on having a BS at the origin does not change the statistical properties of the coexisting PPPs. Hence, the analysis holds for a generic BS located at a generic location.  For the tagged active UE operating on a tagged channel,  the {\rm SINR} experienced at the BS located at the origin can be written as
\beq
{\rm SINR} = \frac{ \rho_o h_o}{\sigma^2+ \underbrace{\sum_{u_i\in \tilde{\mathbf{\Phi}}} P_{i} h_i \left\|u_i\right\|^{-\eta} }_{\mathcal{I}}}
\eeq
where the useful signal power is equal to $\rho_o h_o$ due to the truncated channel inversion power control, $\sigma^2$ is the noise power, and the random variable $\mathcal{I}$ denotes the aggregate interference at the tagged BS from the uplink transmissions by other active UEs on the tagged channel. Note that $\mathcal{I}$ is not identified with the channel index because all channels have i.i.d. interference. The {\rm SINR} outage probability can be calculated as
\begin{align}
\mathbb{P}\left\{{\rm SINR} \leq \theta \right\} &= \mathbb{P} \left\{\rho_o h_o \leq \theta \left(\sigma^2+ \mathcal{I} \right) \right\} \notag \\
 &= \mathbb{E}\left[1-\exp\left\{ -\frac{\theta}{\rho_o} \left(\sigma^2+ \mathcal{I}\right) \right\}\right] \notag \\
&= 1-\exp\left\{ -\frac{\theta}{\rho_o} \sigma^2\right\} \mathcal{L}_{\mathcal{I}}\left(\frac{\theta}{\rho_o}\right)  
\label{SINR}
\end{align}
where the expectation in the second line of (\ref{SINR}) is w.r.t. $\mathcal{I}$, and $\mathcal{L}_{\mathcal{I}}(.)$ denotes the Laplace transform (LT) of the {\em pdf} of the random variable ${\mathcal{I}}.$\footnote{With a slight abuse of terminology, we will denote the LT of the {\em pdf} of a random variable $X$  by the LT of $X$.}

As discussed in Section \ref{intro}, in the uplink, the interfering UEs do not constitute a PPP due to the correlations among them. The correlation among the UEs is due to the unique channel assignment per user in each BS. Hence, the interfering UEs are better modeled using a soft-core process (e.g., Strauss point process~\cite{pp-cellular}) to capture the pairwise correlations among the locations of the active UEs  per channel. Generally, soft-core processes are not analytically tractable \cite{our-survey,martin-book}, and hence, an exact expression for the LT of the aggregate interference cannot be obtained. For this reason, we will approximate the locations of the interfering UEs with a PPP of the same intensity. Note that this approximation only partially ignores the correlations introduced by the system model because the correlation with the tagged BS and the tagged UE is captured by our model. 

\textcolor{red}{It is worth mentioning that the PPP assumption for the interference sources has been widely exploited in the literature (even with the hard core point processes which introduce stronger pairwise correlation between points) and has been proved to be accurate if the correlation among the interfering nodes and the tagged receiver is captured \cite{and-uplink, and-d2d, letter, bac1, Mod-HCPP}. It is also important to note that, due to the correlations among  the shapes of the neighboring Voronoi cells and the fact that the sytem allows only one UE per Voronoi cell, the distances between the scheduled (i.e., interfering) UEs and their serving BSs are identically distributed but are not independent~\cite{and-uplink}. Therefore, as a direct consequence of the employed channel inversion power control, the transmit powers of the scheduled UEs are identically distributed (i.e., follow the distribution derived in \textbf{Lemma~1}) but are not independent. However, the authors in \cite{and-uplink} showed that the dependence among the distances between the scheduled UEs and their serving BSs is weak and can be ignored. Therefore, for analytical tractability, we will assume that the active UEs constitute a PPP and their transmit powers are independent. The accuracy of this assumption will be validated via simulations. Exploiting the PPP approximation and the i.i.d. transmit powers for the set of interfering UEs in the uplink, the outage probability for a generic active UE can be given by the following theorem}.

\vspace{0.2cm}
\begin{theorem}
\label{th1}
In a single-tier Poisson cellular network with truncated channel inversion power control with the cutoff threshold $\rho_o$, in the uplink, assuming that the interfering UEs constitute a PPP and that their transmit powers are independent, the {\rm SINR} outage probability for a generic active UE is given by 
\small
\begin{equation}
\mathcal{O}_s= 1-\exp\left\{ -\frac{\theta \sigma^2}{\rho_o}  -   \frac{2  \theta^\frac{2}{\eta} \gamma\left(2, \pi \lambda \left(\frac{P_u}{\rho_o}\right)^\frac{2}{\eta}\right)}{\left(1-e^{-\pi \lambda \left(\frac{P_u}{\rho_o}\right)^{\frac{2}{\eta}}}\right)} \int_{\theta^\frac{-1}{\eta}}^{\infty} \frac{y}{y^{\eta}+1 }dy \right\}.  
\label{out2}
\end{equation}
\normalsize
%\small
%\begin{equation}
%\mathcal{O}_s= 1-\exp\left\{ -\frac{\theta \sigma^2}{\rho_o}  - {2 \pi   \lambda} {\left(\frac{\theta}{\rho_o}\right)^\frac{2}{\eta}} \mathbb{E}_{P}\left[ P^\frac{2}{\eta}\right] g(\theta,\eta) \right\}  
%\label{out2}
%\end{equation}
%\normalsize
%where $g(\theta,\eta) = \int_{\theta^\frac{-1}{\eta}}^{\infty} \frac{y}{y^{\eta}+1 }dy$ 
\end{theorem}

\begin{proof}
See \textbf{Appendix \ref{proof2}}.
\end{proof}

\vspace{0.2cm}
Generally, it can be shown that the {\rm SINR} outage $\mathcal{O}_s$ is non-increasing in $\rho_o$ (c.f. \ref{out_r}). Recalling (from Section \ref{tx}) that both the truncation outage and the average transmit power are non-decreasing in $\rho_o$, it is concluded that $\rho_o$ introduces a tradeoff in the system performance as will be shown in the numerical results.
 
\begin{figure}[t]
	\begin{center}
		\includegraphics[trim = 3cm 7.8cm 3cm 8cm , clip, width=2.7 in]{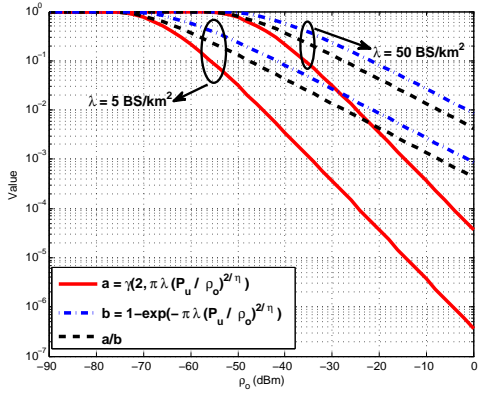}
	  \end{center}
	\caption{ The behavior of $\frac{\gamma(2, \pi \lambda \left(\frac{P_u}{\rho_o}\right)^\frac{2}{\eta})}{\left(1-\exp(- \pi \lambda \left(\frac{P_u}{\rho_o}\right)^\frac{2}{\eta})\right)}$  for $P_u=1$ and $\eta=4$.}
	\label{out_r}
\end{figure}

The {\rm SINR} statistics also controls the average uplink spectral efficiency obtained via Shannon's formula. The average uplink spectral efficiency can be obtained from the following theorem.

\vspace{0.2cm}
\begin{theorem}
\label{th2}
In a single-tier Poisson cellular network with truncated channel inversion power control with the cutoff threshold $\rho_o$, in the uplink, assuming that the interfering UEs constitute a PPP and that their transmit powers are independent, the average spectral efficiency of transmission by an active UE is given by 
\small
\begin{align}
\mathcal{R}&=  \int_{0}^{\infty} \frac{1}{x+1} \exp\left\{ -\frac{x \sigma^2}{\rho_o}  - \right. \notag \\ &\left. \frac{2  x^\frac{2}{\eta} \gamma\left(2, \pi \lambda \left(\frac{P_u}{\rho_o}\right)^\frac{2}{\eta}\right)}{\left(1-e^{-\pi \lambda \left(\frac{P_u}{\rho_o}\right)^{\frac{2}{\eta}}}\right)} \int_{x^\frac{-1}{\eta}}^{\infty} \frac{y}{y^{\eta}+1 }dy \right\}  dx.
\label{rate1}
\end{align}
\normalsize
\end{theorem}

\begin{proof}
See \textbf{Appendix \ref{proof3}}.
\end{proof}

\vspace{0.2cm}
Equation (\ref{rate1}) incorporates the tradeoff between $\rho_o$ and spectral efficiency. On one hand, increasing $\rho_o$ improves the spectral efficiency and decreases the {\rm SINR} outage. On the other hand, increasing $\rho_o$ deteriorates the truncation outage probability and increases the transmit power  of the UEs. In the following, we show some interesting special cases that help understanding the uplink system performance. 

\vspace{0.2cm}
\subsubsection{Special case 1 (infinite $P_u$)} The case of infinite maximum transmit power $P_u$ is of particular interest because it captures the scenario where the transmit power is not a binding constraint for the uplink communication (cf. Fig.\ref{net2}). In other words, the BSs are dense enough (with respect to the required cutoff threshold) such that the distance between a generic UE and its serving BS is relatively small, and hence, the transmit power will be less than the maximum value $P_u$ almost surely. In the analysis, we capture this case by setting $P_u =\infty$. In this case, the {\rm SINR} is given by
\small
\begin{equation}
\mathcal{O}_s= 1-\exp\left\{ -\frac{\theta \sigma^2}{\rho_o}  - 2  \theta^\frac{2}{\eta} \int_{\theta^\frac{-1}{\eta}}^{\infty} \frac{y}{y^{\eta}+1 }dy \right\}  
\label{c1}
\end{equation}
\normalsize
and the spectral efficiency reduces to 
\small
\begin{equation}
\mathcal{R}= \int_0^\infty \frac{\exp\left\{ -\frac{x \sigma^2}{\rho_o}  - 2  x^\frac{2}{\eta} \int_{x^\frac{-1}{\eta}}^{\infty} \frac{y}{y^{\eta}+1 }dy \right\} }{x+1}dx.
\label{cr1}
\end{equation}
\normalsize

Equations (\ref{c1}) and (\ref{cr1}) show that when the transmit power is not a binding constraint, the {\rm SINR} outage and the average spectral efficiency in a single-tier cellular network are independent of the BS intensity. \textcolor{red}{This result is in compliance with the results in \cite{trac} and  it can be concluded that, in terms of {\rm SINR} outage and average spectral efficiency, both the uplink performance (in case of non-binding maximum transmit power constraint) and the downlink performance  are independent of the BS intensity}. 

It is quite insightful to see that the {\rm SINR} outage probability and the average spectral efficiency are independent of the intensity of the BSs. That is, when the maximum transmit power $P_u$ is not a binding operational constraint for the UEs, increasing the intensity (number) of the BSs neither degrades nor improves the {\rm SINR} outage probability and the average spectral efficiency. For the uplink, this behavior can be explained as follows: %as the intensity of the BSs increases, the intensity of interfering UEs increases, %(assuming at least one uplink transmission per BS)
 since each BS will be serving one user per channel, as the intensity of the BSs increases, the intensity of the interfering UEs increases\footnote{This is due to the saturation condition assumption which occurs in a congested network scenario.}; %(assuming at least one uplink transmission per BS)
 however, the average distance between a UE and its serving BS decreases which decreases the transmit power required to maintain the received signal power at $\rho_o$. Hence, the increased intensity of interfering UEs is compensated by the decreased transmit power per interfering UE, and the {\rm SINR} statistics does not change if $\frac{\rho_o}{\sigma^2}$ is kept constant, and vice versa. Therefore, the coverage probability and the average spectral efficiency  can only be improved through interference management techniques such as frequency reuse \cite{ffr1}, interference cancellation \cite{tony-up}, multiple-input-multiple-output (MIMO) antennas \cite{mimo}, interference avoidance via cognition \cite{our-cog2}, or multi-cell cooperation \cite{saquib}. Although these results are only valid for the PPP network model, they are insightful because they reflect the worst-case network performance. More specifically, deploying more BSs, in the worst case, will not degrade the {\rm SINR} statistics.  

\vspace{0.2cm}
\subsubsection{Special case 2 ($\eta=4$)} The integral in (\ref{out2}) can be expressed in closed-form for integer values of $\eta$. For instance, if $\eta=4$, the {\rm SINR} outage probability reduces to the following closed-form:

\small
\begin{equation}
\mathcal{O}_s= 1-\exp\left\{ -\frac{\theta \sigma^2}{\rho_o}  -   \frac{  \sqrt{\theta} \gamma\left(2, \pi \lambda \sqrt{\frac{P_u}{\rho_o}}\right)}{\left(1-e^{-\pi \lambda \sqrt{\frac{P_u}{\rho_o}}}\right)} \arctan(\sqrt{\theta })  \right\}  
\label{c2}
\end{equation}
\normalsize
and the spectral efficiency reduces to the following:
\small
\begin{equation}
\mathcal{R}= \int_0^\infty \frac{\exp\left\{ -\frac{x \sigma^2}{\rho_o}  -   \frac{  \sqrt{x} \gamma\left(2, \pi \lambda \sqrt{\frac{P_u}{\rho_o}}\right)}{\left(1-e^{-\pi \lambda \sqrt{\frac{P_u}{\rho_o}}}\right)} \arctan(\sqrt{x})  \right\}}{x+1}   dx.
\label{cr2}
\end{equation}
\normalsize

\vspace{0.2cm}
\subsubsection{Special case 3 ($\eta=4$, infinite $P_u$, interference-limited scenario)} In the interference-limited case, $\rho_o$ is assumed to be large enough such that $\frac{\theta \sigma^2}{\rho_o} << \mathcal{I}$ and hence noise can be ignored. In this case, the {\rm SINR} outage reduces to the following simple closed-form: 

\small
\begin{equation}
\mathcal{O}_s= 1-\exp\left\{   -  \sqrt{\theta} \arctan(\sqrt{\theta})  \right\}  
\label{c3}
\end{equation}
\normalsize
and the spectral efficiency reduces to
\small
\begin{align}
\mathcal{R} &= \int_{0}^\infty \frac{\exp\left\{   -  \sqrt{x}  \arctan(\sqrt{x})  \right\}}{x+1}  dx \notag \\
&= 0.77 \;\; \text{nat/sec/Hz}.
\label{cr3}
\end{align}
\normalsize

Equation (\ref{c3}) gives a closed-form expression for the outage probability in its simplest form and (\ref{cr3}) shows that the average rate reduces to a constant value\footnote{\textcolor{red}{The correspondence between this result and the result for the downlink case will be discussed in Section.~\ref{dicuss}}.} in this special case. These equations clearly show that when $P_u$ is not a binding constraint and the interference is much larger than the noise, both the {\rm SINR} outage and the spectral efficiency are independent of the cutoff threshold as well as the BS intensity.

\section{Uplink Modeling in Multi-tier Cellular Networks}\label{multi}

In this section, we show that the developed baseline model for uplink transmission in single-tier cellular networks can be extended for uplink transmission in multi-tier cellular networks. First, we present the analysis for multi-tier cellular networks with common path-loss exponent. Then, we will show and comment on the case with different path-loss exponents in the different tiers.

\subsection{Common Path-loss Exponent}

The developed model naturally captures multi-tier cellular networks with different intensities $\lambda_k$ and different cutoff thresholds $\rho_o^{(k)}$ for the different tiers but with a common path-loss exponent $\eta$. In this case, the distribution of the transmit power for  the active UEs is obtained from the following lemma.

\vspace{0.2cm}
\begin{lemma}
\label{lem2}
In a $K$-tier Poisson cellular network with a common path-loss exponent $\eta$ and a truncated channel inversion power control where each tier has the BS intensity of $\lambda_k$ and the cutoff threshold $\rho^{(k)}_o$, the distribution of the transmit power of the active UEs in the uplink in the $j^{th}$ tier is given by
 \begin{align}
f_{P_j}(x) &=  \frac{2 \pi \Lambda x^{\frac{2}{\eta}-1} e^{-\pi \Lambda \left(\frac{x}{\rho^{(j)}_o}\right)^{\frac{2}{\eta}}}}{\eta  (\rho^{(j)}_o)^{\frac{2}{\eta}} \left(1-e^{-\pi \Lambda \left(\frac{P_u}{\rho^{(j)}_o}\right)^{\frac{2}{\eta}}}\right)}, 0\leq x \leq P_u \notag 
\end{align}
where $\Lambda = \sum_{k=1}^K \lambda_k$. The moments of the transmit power of a UE in the $j^{th}$ tier can be obtained as 
\begin{equation}
\mathbb{E}\left[P_j^\alpha\right] =   \frac{(\rho^{(j)}_o)^\alpha \gamma\left(\frac{\alpha \eta}{2}+1, \pi \Lambda \left(\frac{P_u}{\rho^{(j)}_o}\right)^\frac{2}{\eta}\right)}{(\pi \Lambda)^{\frac{\alpha \eta}{2}}\left(1-e^{-\pi \Lambda \left(\frac{P_u}{\rho^{(j)}_o}\right)^{\frac{2}{\eta}}}\right)}.  
\end{equation}
\end{lemma}

\begin{proof}
 Since the UEs associate to the BSs based on the average uplink link quality, following the superposition theorem of the PPP \cite{martin-book}, the association regions for the BSs form a Voronoi tessellation for a PPP with intensity $\Lambda$. Note that, for the uplink in a multi-tier network, the association regions form a Voronoi tessellation rather than a weighted Voronoi tessellation due to the homogenous transmit powers of the UEs. The rest of the proof follows the same steps as in \textbf{Appendix \ref{proof1}}.
\end{proof}

\vspace{0.2cm}
The truncation outage probability at the $j^{th}$ tier is given by
\begin{equation} 
\mathcal{O}^{(j)}_p = e^{-\pi \Lambda \left(\frac{P_u}{\rho^{(j)}_o}\right)^{\frac{2}{\eta}}}.
\label{op2}
\end{equation} 

Without loss of generality, let {\rm SINR$_j$} be the {\rm SINR} experienced by a tagged BS in the $j^{th}$ tier and the BS is located at the origin. By  Slivnyak's theorem, the analysis holds for any BS in the $j^{th}$ tier. Let each tier has its own {\rm SINR} threshold $\theta_k$, then, following (\ref{SINR}), the {\rm SINR} outage in the $j^{th}$ tier can be expressed as

\begin{align}
\mathbb{P}\left\{{\rm SINR_j} \leq \theta_j \right\} &= \mathbb{E}\left[1-\exp\left\{ -\frac{\theta_j}{\rho^{(j)}_o} \left(\sigma^2 + \sum_{k=1}^K {\mathcal{I}_k}\right)\right\}\right]   \notag \\
&= 1-\exp\left\{ -\frac{\theta_j}{\rho^{(j)}_o} \sigma^2\right\} \prod_{k=1}^K \mathcal{L}_{\mathcal{I}_k}\left(\frac{\theta_j}{\rho^{(j)}_o}\right)  
\label{SINR3}
\end{align}
where $\mathcal{I}_k$ is the aggregate interference from the $k^{th}$ tier, and the second line of equation (\ref{SINR3}) follows from the independence of $\mathcal{I}_k$ and $\mathcal{I}_i$ $\forall i \neq k$. Note that $\mathcal{I}_j$ represents the co-tier interference and $\mathcal{I}_k$ for $k \neq j$ represents the cross-tier interference. \textcolor{red}{As discussed in the single-tier case, the interfering UEs do not constitute a PPP and their transmit powers are correlated. However, for analytical tractability, similar to the single-tier case, we approximate the interfering UEs in each tier with a PPP and ignore the transmit power correlations}. The {\rm SINR} outage for a generic active UE in the $j^{th}$ tier is then given by the following theorem.

\vspace{0.2cm}
\begin{theorem}
\label{th3}
In a $K$-tier Poisson cellular network with a common path-loss exponent $\eta$ and truncated channel inversion power control where each tier has BS intensity $\lambda_k$ and  cutoff threshold $\rho^{(k)}_o$, in the uplink, assuming that the interfering UEs in each tier constitute an independent PPP and that their transmit powers are independent, the {\rm SINR} outage probability for a generic active UE  in the $j^{th}$ tier is given by 
\small
\begin{align}
\mathcal{O}^{(j)}_s &= 1-\exp\left\{ -\frac{\theta_j \sigma^2}{\rho^{(j)}_o}  -\sum_{k=1}^K  \left(\frac{\theta_j  (\rho^{(k)}_o)}{\rho^{(j)}_o}\right)^\frac{2}{\eta} \right. \notag \\ & \left. \frac{2  \lambda_k  \gamma\left(2, \pi \Lambda \left(\frac{P_u}{\rho^{(k)}_o}\right)^\frac{2}{\eta}\right)}{\Lambda \left(1-e^{-\pi \Lambda \left(\frac{P_u}{\rho^{(k)}_o}\right)^{\frac{2}{\eta}}}\right)} \int_{\left(\frac{\theta_j \rho^{(k)}_o}{\rho^{(j)}_o}\right)^\frac{-1}{\eta}}^{\infty} \frac{y}{y^{\eta}+1 }dy \right\}. 
\label{out3}
\end{align}
\normalsize
\end{theorem}

\begin{proof}
See \textbf{Appendix \ref{proofth3}}.
\end{proof}

\vspace{0.2cm}
The {\rm SINR} outage in (\ref{out3}) reduces to the single-tier case given in (\ref{out2}), but with intensity $\Lambda$, when all tiers have the same cutoff threshold $\rho_o$ despite of their 
%different downlink transmit powers and their 
different intensities. The spectral efficiency in the multi-tier case can be obtained from the following theorem.

\vspace{0.2cm}
\begin{theorem}
\label{th4}
In a $K$-tier Poisson cellular network with a common path-loss exponent $\eta$ and truncated channel inversion power control where each tier has BS intensity $\lambda_k$ and cutoff threshold $\rho^{(k)}_o$,  in the uplink, assuming that the interfering UEs in each tier constitute a independent PPP and that their transmit powers are independent, the average spectral efficiency of transmission by an active UE in the $j^{th}$ tier is given by  
\small
\begin{align}
\mathcal{R}_j&=  \int_{0}^{\infty} \frac{1}{x+1} \exp\left\{ -\frac{x \sigma^2}{\rho^{(j)}_o}  -  \sum_{k=1}^K \left(\frac{x \rho^{(k)}_o }{\rho^{(j)}_o}\right)^\frac{2}{\eta} \right. \notag \\ & \left. \frac{2 \lambda_k  \gamma\left(2, \pi \Lambda \left(\frac{P_u}{\rho^{(k)}_o}\right)^\frac{2}{\eta}\right)}{\Lambda \left(1-e^{-\pi \Lambda \left(\frac{P_u}{\rho^{(k)}_o}\right)^{\frac{2}{\eta}}}\right)} \int_{\left(\frac{x \rho^{(k)}_o}{\rho^{(j)}_o}\right)^\frac{-1}{\eta}}^{\infty} \frac{y}{y^{\eta}+1 }dy \right\}   dx.
\label{rate2}
\end{align}
\normalsize
\end{theorem}

\begin{proof}
The proof is similar to the one in \textbf{Appendix \ref{proof3}}.
\end{proof}

\vspace{0.2cm}
For brevity, we consider only one special case. For infinite $P_u$, interference-limited network scenario, and $\eta=4$, the {\rm SINR} outage and average spectral efficiency in the uplink in the $j^{th}$ tier reduce to the following:

\small
\begin{align}
\mathcal{O}^{(j)}_s &= 1-\exp\left\{- \sum_{k=1}^K  \frac{ \lambda_k}{\Lambda}  \sqrt{\frac{\theta_j \rho^{(k)}_o}{\rho^{(j)}_o}}\arctan\left({\sqrt{\frac{\theta_j \rho^{(k)}_o}{\rho^{(j)}_o}}}\right) \right\}  
\label{mc1}
\end{align}
\normalsize

\small
\begin{align}
\mathcal{R}_j&=  \int_{0}^{\infty} \frac{\exp\left\{- \sum_{k=1}^K  \frac{ \lambda_k}{\Lambda}  \sqrt{\frac{x \rho^{(k)}_o}{\rho^{(j)}_o}}\arctan\left({\sqrt{\frac{x \rho^{(k)}_o}{\rho^{(j)}_o}}}\right)  \right\}}{x+1} dx.
\label{mcr1}
\end{align}
\normalsize

Equations (\ref{mc1}) and (\ref{mcr1}) show that, in general, the {\rm SINR} outage probability and spectral efficiency in a certain tier depend on the relative cutoff thresholds and the relative BS intensities. The {\rm SINR} outage probability decreases and the spectral efficiency improves as the cutoff threshold in the target tier increases and the cutoff thresholds in the other tiers and/or the BS intensities in the other tiers decrease. This is because, a higher cutoff threshold in the target tier increases the useful signal power and lower cutoff thresholds and/or lower BS intensities in the other tiers reduce the cross-tier interference. 

\textbf{Theorems~\ref{th3}} and \textbf{\ref{th4}} show an important difference between multi-tier and single-tier cellular networks. That is, regardless of the maximum transmit power value $P_u$ (i.e., binding or non-binding maximum transmit power constraint), the {\rm SINR} outage probability and spectral efficiency in multi-tier cellular networks depend on the relative values of the cutoff thresholds and the relative BS intensities. However, note that if all the tiers have the same cutoff threshold, regardless of the relative BSs intensities,  for an interference-limited scenario with $\eta=4$ and infinite $P_u$, (\ref{mc1}) and (\ref{mcr1}) reduce to (\ref{c3}) and (\ref{cr3}), respectively. Hence, the multi-tier cellular network can be reduced to a single-tier network with intensity $\Lambda$, and both the {\rm SINR} outage probability and the average spectral efficiency become independent of the BS intensities in the different tiers  if $P_u$ is non-binding. It is worth mentioning that similar performance  was observed in \cite{trac2} for the downlink multi-tier case. \textcolor{red}{In particular, it was shown that if each UE associates to the BS offering the highest average received power, the outage probability and spectral efficiency become independent of the relative intensities  and transmit powers of BSs in different network tiers}. %Hence, in the non-binding transmit power case, common path-loss, and common cutoff threshold $\rho_o$ the uplink and downlink performance matches and become independent form relative tiers intensities.    

%It is worth mentioning that in the downlink multi-tier case, it was shown in \cite{trac2} that, if the UEs are associating to the BS offering the highest average received power, the outage probability and spectral efficiency become independent of the relative tiers intensities and transmit powers. Hence, in the non-binding transmit power case, common path-loss, and common cutoff threshold $\rho_o$ the uplink and downlink performance matches and become independent form relative tiers intensities.    

\subsection{Different Path-loss Exponents}

In the previous section, it has been shown that when all network tiers share the same path-loss exponent, the association of the UEs does not change from that in the single-tier cellular networks. That is, the association regions of the BSs will form a Voronoi tessellation. On the other hand, if different tiers have different path-loss exponents, the association of UEs in the multi-tier case deviates from that in the single-tier case where the association regions of the BSs form a weighted Voronoi tessellation. That is, the BSs in tiers with lower path-loss exponents will have larger service areas than BSs in tiers with higher path-loss exponents. The transmit power in a multi-tier cellular network with different path-loss exponents can be characterized by the following lemma. 

\vspace{0.2cm}
\begin{lemma}
\label{lem5}
In a $K$-tier Poisson cellular network with truncated channel inversion power control where each tier has BS intensity $\lambda_k$, cutoff threshold $\rho^{(k)}_o$, and path-loss exponent $\eta_k$, the {\em pdf} of the transmit power of the active UEs in the uplink in the $j^{th}$ tier is given by
\begin{align}
f_{P_j}(x) &=   \frac{\sum_{k=1}^{K} \frac{2 \pi \lambda_k x^{\frac{2}{\eta_k}-1} }{\eta_k (\rho_o^{(j)})^\frac{2}{\eta_k} }}{1-e^{-\sum_{b=1}^{K}\pi \lambda_b \left(\frac{P_u}{\rho^{(j)}_o}\right)^{\frac{2}{\eta_b}}}} {e^{-\sum_{a=1}^{K}\pi \lambda_a \left(\frac{x}{\rho_o^{(j)}}\right)^{\frac{2}{\eta_a}}}}.
\end{align}
The moments of the  transmit power of a UE in the $k^{th}$ tier can be obtained as 
\begin{equation}
\mathbb{E}\left[P_j^\alpha\right] =  \int_0^{P_u} \frac{\sum_{k=1}^{K} \frac{2 \pi \lambda_k x^{\frac{2}{\eta_k}+\alpha-1} }{\eta_k (\rho_o^{(j)})^\frac{2}{\eta_k} }}{1-e^{-\sum_{b=1}^{K}\pi \lambda_b \left(\frac{P_u}{\rho^{(j)}_o}\right)^{\frac{2}{\eta_b}}}} {e^{-\sum_{a=1}^{K}\pi \lambda_a \left(\frac{x}{\rho_o^{(j)}}\right)^{\frac{2}{\eta_a}}}}.  
\label{differ}
\end{equation}
\end{lemma}
\begin{proof}
See \textbf{Appendix \ref{lemm5}}.
\end{proof}

\vspace{0.2cm}
The truncation outage probability in the uplink for a UE at the $j^{th}$ tier is given by
\begin{equation} 
\mathcal{O}^{(j)}_p = e^{-\sum_{k=1}^{K}\pi \lambda_k \left(\frac{P_u}{\rho^{(k)}_o}\right)^{\frac{2}{\eta_k}}}.
\label{op3}
\end{equation} 

From \textbf{Lemma~\ref{lem5}} we can see that the moments of the transmit power of the UEs cannot be obtained in closed-form due to the complications introduced by the different path-loss exponents. Note that, since the transmit power appears in the LT of the interference as $\mathbb{E}[P_k^{\frac{2}{\eta_k}}]$, the calculations of outage probability and spectral efficiency  do not require obtaining the moments of the transmit power in closed-forms. It can be shown that, for a common path-loss exponent, \textbf{Lemma~\ref{lem5}} reduces to \textbf{Lemma~\ref{lem2}}. Similar to the previous cases, the interfering UEs do not constitute a PPP and their transmit powers are correlated. However, for analytical tractability, we will approximate the locations of the interfering UEs by a PPP and ignore the correlations in transmit power. The {\rm SINR} outage in multi-tier cellular networks with different path-loss exponents can be characterized by the following theorem. 

\vspace{0.2cm}
\begin{theorem}
\label{th5}
In a $K$-tier Poisson cellular network with truncated channel inversion power control where each tier has BS intensity $\lambda_k$, cutoff threshold $\rho^{(k)}_o$, and path-loss exponent $\eta_k$, assuming that the interfering UEs in each tier constitute an independent PPP and that their transmit powers are independent, the {\rm SINR} outage probability in the uplink for a generic active UE in the $j^{th}$ tier is given by 
\small
\begin{align}
\mathcal{O}^{(j)}_s &= 1-\exp\left\{ -\frac{\theta_j \sigma^2}{\rho^{(j)}_o}  - \sum_{k=1}^K  2 \pi  \lambda_k \left(\frac{\theta_j}{\rho^{(j)}_o}\right)^\frac{2}{\eta_j}   \right. \notag \\ &\left. \quad \quad \quad \quad \mathbb{E}_{P_k}\left[ P_k^{\frac{2}{\eta_j}}\right] \int_{\left(\frac{\theta_j {\rho^{(k)}_o}}{\rho^{(j)}_o}\right)^{\frac{-1}{\eta_j}}}^{\infty} \frac{y}{y^{\eta_j}+1 }dy \right\}.
\label{out4}
\end{align}
\normalsize
\end{theorem}
\begin{proof}
See \textbf{Appendix \ref{proof5}}.
\end{proof}

\vspace{0.2cm}
Note that $\mathbb{E}_{P_k}\left[ P_k^{\frac{2}{\eta_j}}\right]$ is just a number that can be obtained numerically from (\ref{differ}). Hence, (\ref{out4}), as shown earlier, can be obtained in closed-form for integer values of $\eta_j$. The spectral efficiency can be characterized via the following theorem.

\vspace{0.2cm}
\begin{theorem}
\label{th6}
In a $K$-tier Poisson cellular network with truncated channel inversion power control where each tier has the BS intensity $\lambda_k$, cutoff threshold $\rho^{(k)}_o$, and path-loss exponent $\eta_k$, assuming that the interfering UEs in each tier constitute an independent PPP and that their transmit powers are independent, the average spectral efficiency of transmission in the uplink by an active UE in the $j^{th}$ tier is given by 
\small
\begin{align}
\mathcal{R}_j&=  \int_{0}^{\infty} \frac{1}{x+1} \exp\left\{ -\frac{\theta_j \sigma^2}{\rho^{(j)}_o}  - \sum_{k=1}^K  2 \pi  \lambda_k \left(\frac{x}{\rho^{(j)}_o}\right)^\frac{2}{\eta_j}   \right. \notag \\ &\left. \quad \quad \quad \quad \mathbb{E}_{P_k}\left[ P_k^{\frac{2}{\eta_k}} \right] \int_{\left(\frac{x {\rho^{(k)}_o}}{\rho^{(j)}_o}\right)^{\frac{-1}{\eta_j}}}^{\infty} \frac{y}{y^{\eta_j}+1 }dy  \right\}   dx.
\label{rate4}
\end{align}
\normalsize
\end{theorem}
\begin{proof}
The proof is similar to that in \textbf{Appendix \ref{proof3}}. 
\end{proof}
\vspace{0.2cm}
Due to the complicated expressions for the moments of the transmit power as given in \textbf{Lemma~\ref{lem5}}, \textbf{Theorems~\ref{th5}} and~\textbf{\ref{th6}} do not give simple formulas for the {\rm SINR} outage and spectral efficiency. However, \textbf{Theorems~\ref{th5}} and~\textbf{\ref{th6}} do give general formulas that reduce to all of the previously presented special cases. Note that without simple expressions, significant insights may not be extracted from the results obtained in \textbf{Theorems~\ref{th5}} and ~\textbf{\ref{th6}} and \textbf{Lemma~\ref{lem5}} for multi-tier networks with different path-loss exponents. In contrast, with a common path-loss exponent, \textbf{Theorems~\ref{th5}} and ~\textbf{\ref{th6}} and \textbf{Lemma~\ref{lem5}} simplify to  \textbf{Theorems~\ref{th3}} and ~\textbf{\ref{th4}} and \textbf{Lemma~\ref{lem2}}, respectively. The main conclusion from this section is that the developed paradigm is general and flexible to capture different practical system parameters.

%It is worth mentioning that the LT given in \textbf{Appendix~\ref{proof4}} easily generalizes to the different path-loss exponent case. Since our main objective in this paper is to develop a tractable paradigm for the uplink cellular network and provide insights via the simple equation to the uplink operation, we postpone the multi-tier cellular network modeling with different path-loss exponents to our future work.

\section{Results and Discussions} \label{resss}

\subsection{Results}

In this section, we validate our model against simulations and present some numerical results for a single-tier cellular network (or equivalently, a multi-tier cellular network with common cutoff threshold $\rho_o$ and path-loss exponent $\eta$). Unless otherwise stated, we set the BS intensity to $\lambda=2$ BSs/km$^2$, the maximum transmit power $P_u=1$ W, the receiver sensitivity $\rho_{min} = -90$ dBm, the cutoff threshold $\rho=-70$ dBm, the {\rm SINR} threshold $\theta=1$, $\sigma^2 = -90$ dBm, and the number of channels $\left|\mathbf{S}\right|=1$. In the simulation setup, we realize a PPP cellular network with intensity $\lambda$ in a 400 km$^2$ area. Then we drop the UEs randomly and uniformly over the simulation area until all BSs are active (i.e., each BS has a scheduled UE for which $P_u$ is sufficient to invert its path-loss and communicate in the uplink) to ensure that the saturation condition is satisfied. All UEs employ the channel inversion power control. The simulation is repeated 10000 times. 

 First, we validate our model against simulation results obtained for a Poisson cellular network and compare it with the circular approximation of the target BS coverage  as used in \cite{spec-share, and-d2d}. The reason for the model validation is that the derived {\em cdf} of the {\rm SINR} assumes that the set of active UEs constitutes a PPP and that their transmit powers are independent. Note that our proposed analysis captures the correlations between the location of the tagged BS  and the locations of the interfering UEs as well as the correlations between the transmit power of the tagged UE  and the transmit powers of the interfering UEs (cf. \textbf{Facts $\#1$--$\#3$} in \textbf{Appendix~\ref{proof2}}). Hence, our model partially ignores the correlations imposed by the uplink system model. 

%Note that the PPP assumptions partially ignores the correlation among the simultaneously active UEs on the same channel. That is, given that a UE is transmitting on a tagged channel, all other active UEs should be outside the Voronoi cell of its serving BS. Note that our model partially ignores these correlations because the correlation with the tagged BS is captured by the fact that the average received interference power from any interfering UE is less than the cutoff threshold $\rho_o$ (cf. \textbf{Facts $\#1$--$\#3$} in \textbf{Appendix~\ref{proof2}}). 
 
 Fig.~\ref{SINR1} shows that the derived model accurately captures the {\rm SINR} outage. The figure also shows that approximating the coverage area of the tagged BS with a circle with radius $\frac{1}{\sqrt{\pi \lambda}}$ underestimates the outage at $\rho_o \leq -70$ dBm and overestimates the outage at  $\rho_o = -90$ dBm. The reason is that at a low cutoff threshold $\rho_o$, the interference exclusion region around each BS is large (cf. \textbf{Fact $\#2$} in \textbf{Appendix~\ref{proof2}}), and hence, the radius $\frac{1}{\sqrt{\pi \lambda}}$ estimates a more aggressive interference. On the other hand, for a high cutoff threshold $\rho_o$, the interference exclusion region around the tagged BS is small, and hence, the radius $\frac{1}{\sqrt{\pi \lambda}}$ estimates a more conservative interference.

\begin{figure}[t]
	\begin{center}
\includegraphics[width=0.35 \textwidth]{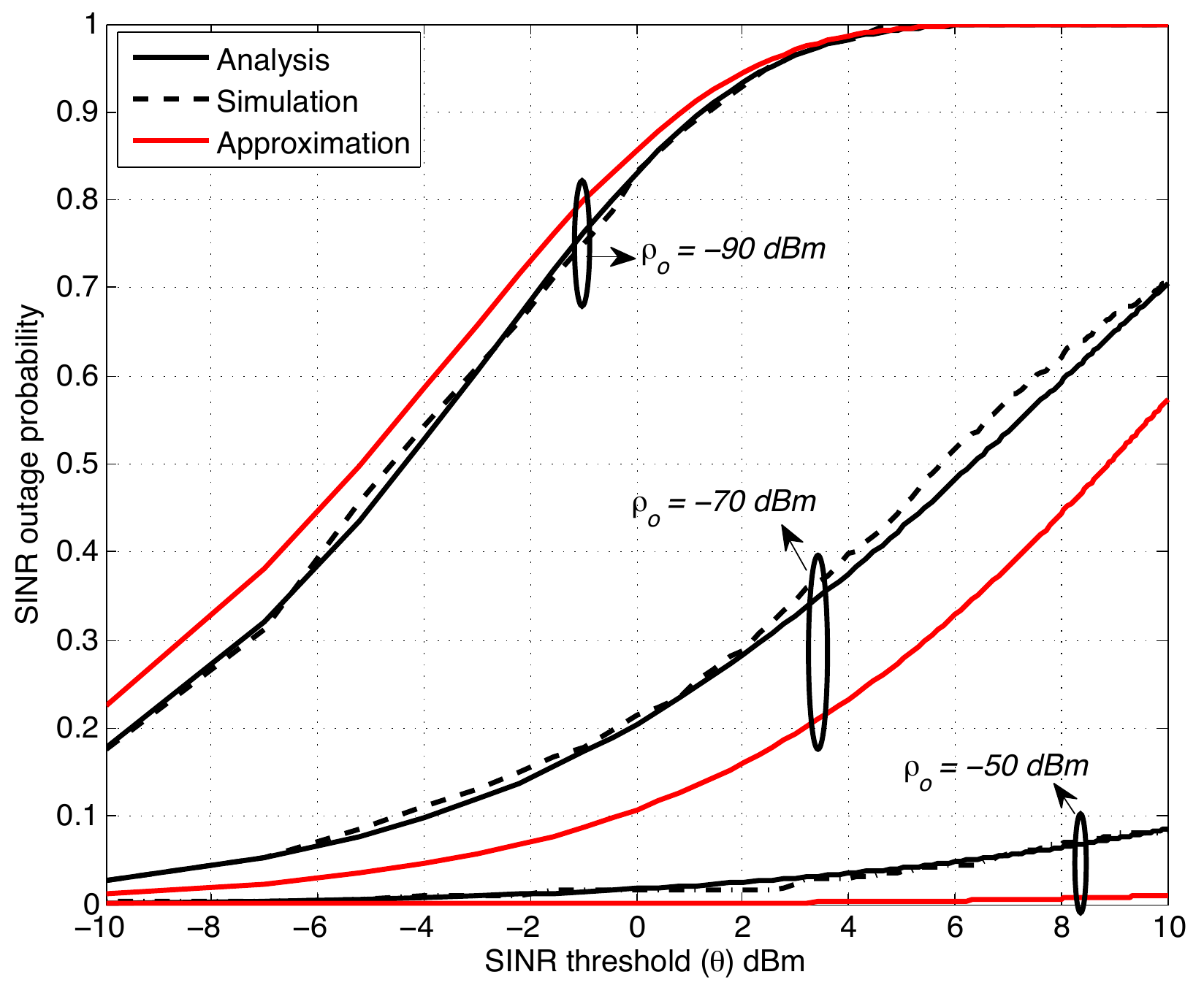}
	\end{center}
	\caption{The {\em cdf} of the {\rm SINR} for $\lambda=2$ BS/km$^2$ and $\left|\mathbf{S}\right|=1$.}
	\label{SINR1}
\end{figure}

Fig.~\ref{SINR2} shows the tradeoff introduced by $\rho_o$ on the total outage probability $\mathcal{O}_t = \mathcal{O}_p + (1-\mathcal{O}_p) \mathcal{O}_s$. As shown in the figure, $\rho_o$ tunes the tradeoff between the two outage probabilities and there exists an optimal cutoff threshold $\rho^{*}_o$ that minimizes the total outage probability. That is, at lower values of $\rho_o$, the {\rm SINR} outage dominates the outage probability. On the other hand, at high values of $\rho_o$, the truncation outage dominates the outage probability. The figure shows the two regions of operation for the uplink, namely, when $P_u$ is a binding constraint and when $P_u$ is a non-binding constraint. Note that $P_u$ becomes binding when the truncation outage probability  is non-zero. For small values of $\rho_o$, $P_u$ induces a non-binding constraint, and hence, for relatively high BS intensity (e.g., $\lambda = 10$ and $\lambda = 100$) the {\rm SINR} outage is independent of the BS intensity (i.e., the two curves for  $\lambda = 10$ and $\lambda = 100$ coincides). This reinforces case $\#1$ in Section \ref{outage}. Note that when $\rho_o$ is comparable to the noise power $\sigma^2$, the {\rm SINR} outage depends on $\rho_o$. However, when $\rho_o$ is much greater than the noise power and $P_u$ is non-binding, the {\rm SINR} outage is independent of both the cutoff threshold $\rho_o$ and BS intensity (case $\#3$ in Section \ref{outage}). For high values of $\rho_o$, $P_u$ becomes a binding constraint and the {\rm SINR} outage depends on both the BS intensity and cutoff threshold. 

%Note that in the case where the outage probability is independent of $\rho_o$, the optimal cutoff threshold $\rho^{*}_o$ is the lowest which minimizes the power consumption in the network as shown in Fig.~\ref{Epc}.  

\begin{figure}[t]
	\begin{center}
\includegraphics[ width=0.35 \textwidth]{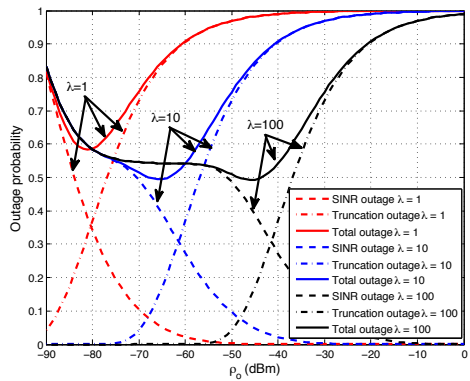}
	\end{center}
	\caption{ Total outage probability for $\theta=1$ and $\left|\mathbf{S}\right|=1$.}
	\label{SINR2}
\end{figure}

To show the two regions of operation for the uplink more clearly, we plot Fig.~\ref{low_noise} with a negligible noise power value of  $\sigma^2=-110$ dBm. This figure shows that when $P_u$ is non-binding (i.e., the truncation outage probability is almost zero), the {\rm SINR} outage is completely independent of both the cutoff threshold $\rho_o$ and BS intensity (case $\#3$ in Section \ref{outage}). Note that the independence w.r.t. the BS intensity can be seen from the coincidence of the curves for $\lambda=10$ and $\lambda =100$ as long as $P_u$ is non-binding for both intensities (i.e., the truncation outage probability is equal to zero). In contrast, when $P_u$ becomes binding (i.e., the truncation outage probability is non-zero), the {\rm SINR} outage is highly affected by both the BS intensity and cutoff threshold (case $\#2$ in Section \ref{outage}).

\begin{figure}[t]
	\begin{center}
\includegraphics[ width=0.35 \textwidth]{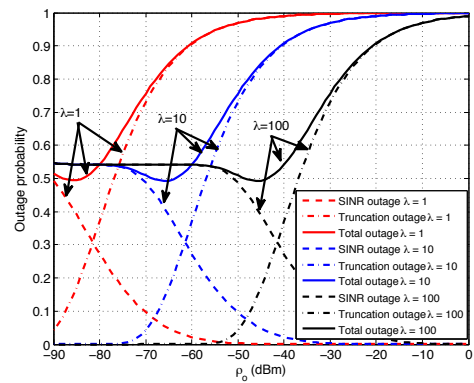}
	\end{center}
	\caption{ Total outage probability for $\sigma^2=-110$ dBm, $\theta=1$, and $\left|\mathbf{S}\right|=1$.}
	\label{low_noise}
\end{figure}

Fig.~\ref{spec} shows the tradeoff introduced by $\rho_o$ on the effective spectral efficiency defined as $(1-\mathcal{O}_p) \mathcal{R}$. The effective spectral efficiency captures the average spectral efficiency for active users (i.e., users with no truncation outage).  As shown in this figure, when $P_u$ is a non-binding constraint (i.e., for $\rho_o \leq -75$ with $\lambda = 10$ and $\rho_o \leq -55$ with $\lambda=100$  [cf. Fig.~\ref{SINR2}]), the effective spectral efficiency is independent of the BS intensity. Note that the effective spectral efficiency depends on $\rho_o$ when $P_u$ induces a non-binding constraint due to the relatively (i.e., relative to $\rho_o$) high noise power  (case $\#1$ in Section \ref{outage}). However, for $\lambda = 100$ the effective spectral efficiency is independent of the cutoff threshold in the range of $-70 \leq \rho_o \leq -55$ dBm because $P_u$ induces a non-binding constraint and the noise power is negligible w.r.t. the value of $\rho_o$  (case $\#3$ in Section \ref{outage}). On the other hand, when $P_u$ becomes a binding constraint, the effective spectral efficiency depends on both the BS intensity and cutoff threshold. This figure also shows the existence of an optimal cutoff threshold $\rho^*_o$ which maximizes the effective spectral efficiency.

\begin{figure}[t]
	\begin{center}
\includegraphics[width=0.35 \textwidth]{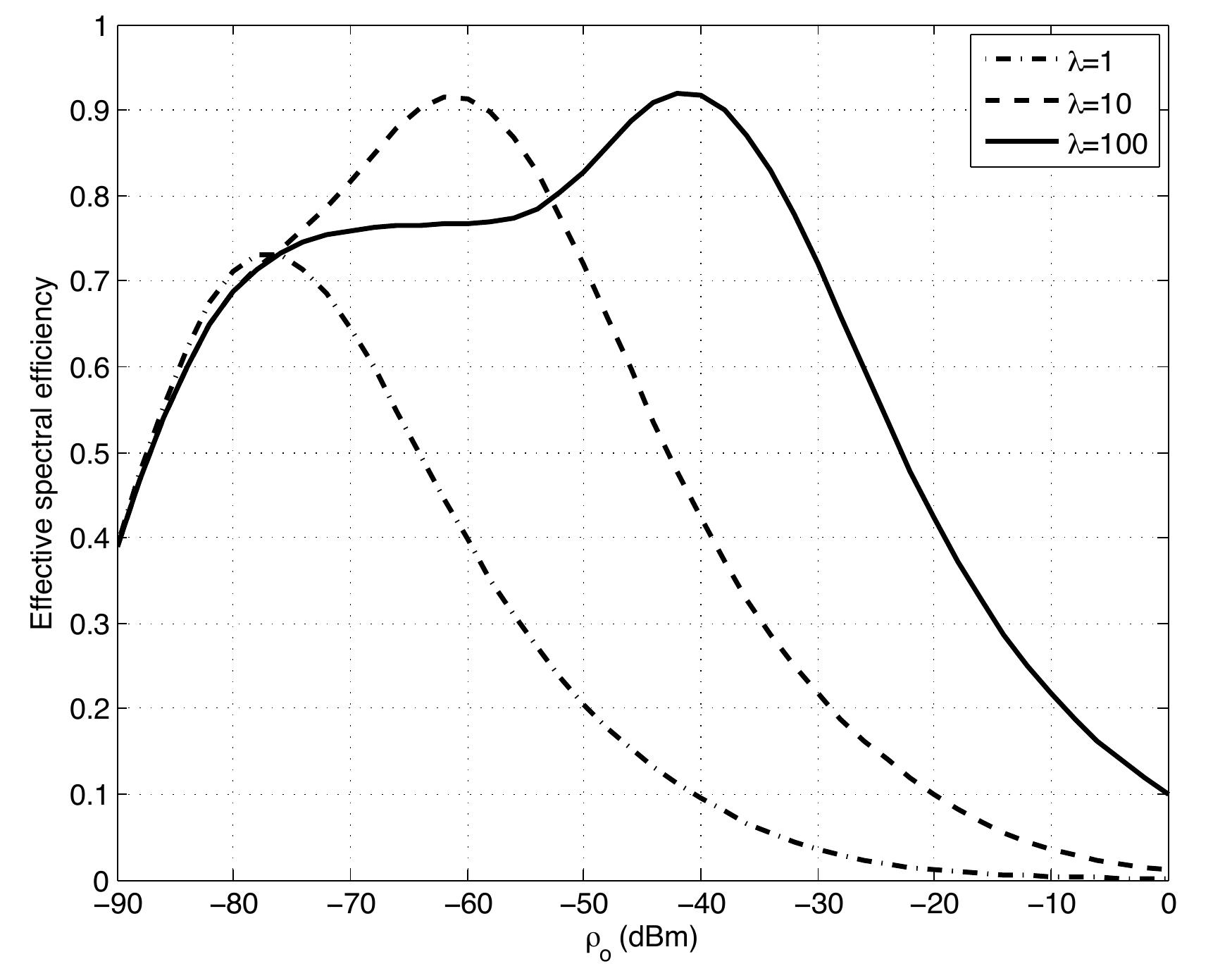}
	\end{center}
	\caption{  Effective spectral efficiency $\sigma^2 = -90$ dBm and $\left|\mathbf{S}\right|=1$.}
	\label{spec}
\end{figure}

To show the tradeoff introduced by $\rho_o$ on the effective spectral efficiency more clearly, we plot Fig.~\ref{spec_noise} for negligible noise power $\sigma^2= -110$ dBm. This figure shows that when $P_u$ induces a non-binding constraint, the effective spectral efficiency is completely independent of the BS intensity and cutoff threshold. On the other hand, when $P_u$ induces a binding constraint, the effective spectral efficiency depends on both the BS intensity and the cutoff threshold. 
%Note that in Figs.~\ref{SINR2}, \ref{low_noise}, ~\ref{spec}, and~\ref{spec_noise} we make the ratio of BS intensity and UE intensity constant to avoid the effect of $\delta$ (which depends on the ratio $\frac{\mathcal{U}}{\lambda}$) so that we see the tradeoff introduced by $\rho_o$ more clearly.

\begin{figure}[t]
	\begin{center}
\includegraphics[width=0.35 \textwidth]{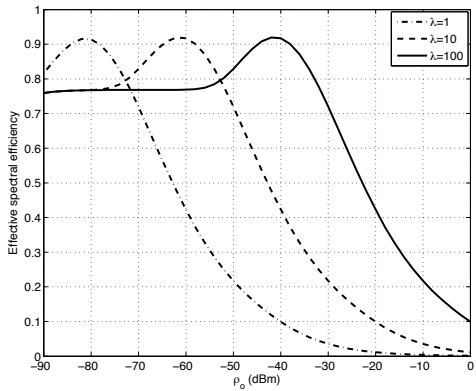}
	\end{center}
	\caption{  Effective spectral efficiency for $\sigma^2 = -110$ dBm, $\frac{\mathcal{U}}{\lambda}=50$, and $\left|\mathbf{S}\right|=1$.}
	\label{spec_noise}
\end{figure}

 Fig.~\ref{Epc} shows the average transmit power of the UEs vs. the cutoff threshold $\rho_o$. As the cutoff threshold $\rho_o$ increases, the UEs are required to transmit at higher powers to invert the path-loss and maintain a high threshold $\rho_o$ at their serving BSs. Therefore, the average transmit power is non-decreasing in $\rho_o$. Note that the average transmit power saturates at $\underset{\rho_o \rightarrow \infty}\lim \mathbb{E}[P] = \frac{P_u}{3}$.

\begin{figure}[t]
	\begin{center}
\includegraphics[width=0.35 \textwidth]{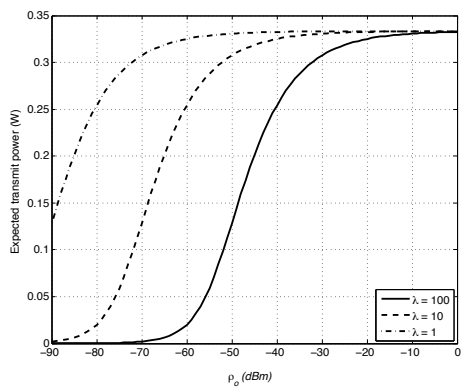}
	\end{center}
	\caption{ Expected transmit power in the uplink.}
	\label{Epc}
\end{figure}

\subsection{Discussions}\label{dicuss}

In the light of the proposed uplink framework and the results provided in \cite{trac, trac2}, we highlight the commonalities and differences between the uplink and downlink transmission performances. The criticality of the power control is the first main difference between the uplink and downlink transmission performances. While power control enhances the downlink transmission performance, it is not crucial for the basic network operation due to the inherent interference protection provided by the user association policy. On the other hand, power control is essential for the case of uplink operation to mitigate potential severe interference caused by the arbitrary close interfering UEs (cf. Fig.~\ref{per_chan}). 

\textcolor{red}{The second difference between the uplink and downlink operation is the maximum transmit power constraint for the UEs. For a single-tier cellular network, the relative values of $\lambda$, $P_u$, and $\rho_o$ may lead to a binding maximum transmit power constraint for the uplink operation. In this case, the uplink performance is highly affected by the intensity of the BSs as well as the cutoff threshold for power control. On the other hand, when the relative values of $\lambda$, $P_u$, and $\rho_o$ lead to an uplink operation under non-binding maximum transmit power constraint, the uplink transmission performance is analogous to the downlink transmission performance. More specifically, the uplink network can be transformed into an equivalent network where all UEs have a constant transmit power of $E\left[P^{\frac{2}{\eta}}\right]$ (see (\ref{lt1}))\footnote{\textcolor{red}{The expression for LT of the interference  in  (\ref{lt1}) shows the equivalence of the interference distributions in the uplink and downlink transmission scenarios where the randomness in the transmit  power is captured via its fractional moments in the expression for LT  as discussed in \cite{req1, req2}.}}% Note that this equivalence is to the interference protection imposed by the power control. However, the expressions for the {\rm SINR} outage and spectral efficiency obtained in this paper show the difference between the uplink and downlink models (when compared to the expressions in \cite{trac, trac2}) due the the fixed power level maintained by the power control mechanism at the designated receivers.}}. 
. Hence, the uplink performance with non-binding $P_u$ is similar to the downlink performance with BS transmit power of $E\left[P^{\frac{2}{\eta}}\right]$ and a constant downlink distance of $\left(\frac{E[P^{\frac{2}{\eta}}]}{\rho_o}\right)^\frac{1}{\eta}$. It is important to note that the similarity between the uplink performance (with non-binding $P_u$) and the downlink performance is due to the power control in the uplink which enforces an interference protection for the BSs. This is analogous to the interference protection imposed by the association in the downlink for the UEs. However, the power control maintains a fixed power level of $\rho_o$ at the BSs which increases the {\rm SINR} outage probability (e.g., as seen by comparing (\ref{c3}) to the downlink outage given by $1- \left(1+\sqrt{\theta}\arctan{{\sqrt{\theta}}}\right)^{-1}$ \cite{trac}), and reduces the spectral efficiency (e.g., $0.77$ nat/sec/Hz obtained in (\ref{cr3}) for the uplink when compared to the $1.49$ nat/sec/Hz  obtained for the downlink in~\cite{trac}). A better performance is achieved in the downlink because  the UEs which are close to their serving BSs experience high received power (w.r.t. the fixed $\rho_o$ in the uplink) from their serving BSs. To this end, it can be concluded that, in a single-tier cellular network with a loose maximum transmit power constraint, the {\rm SINR} outage and spectral efficiency performances for uplink communication  are independent of the BS intensity. Therefore, in the worst case (due to the PPP assumption), deploying more BSs will not affect the {\rm SINR} statistics in either the uplink or the downlink.}

%the uplink performance becomes independent of the relative BS intensities in the case of a loose maximum transmit power constraint. In this case, as discussed in \cite{trac2} in the context of  downlink communication, the results in this paper  show that

For multi-tier cellular networks,  regardless of whether $P_u$ is binding or not, the  uplink performance depends on the relative cutoff thresholds and relative BS intensities in the network tiers. However, in a multi-tier cellular network  with a common cutoff threshold, a common path-loss exponent, and a loose maximum transmit power constraint,  the {\rm SINR} outage and the spectral efficiency performances for uplink communications  are independent of the BS intensity. Therefore, in the worst case  (because of the PPP assumption), deploying more BSs or more tiers will not affect the {\rm SINR} statistics in either the uplink or the downlink. 

%Another important difference between the uplink and the downlink is the UE association. The heterogeneity in the downlink transmit powers of BSs in a multi-tier cellular network leads to a coverage in the form of a weighted Voronoi tessellation. On the other hand, in the uplink, the homogeneity of the transmit powers of the UEs leads to a coverage in the form of Voronoi tessellation in a multi-tier cellular network with a common path-loss exponent. The simple association in the uplink highly simplifies the analysis in the case of multi-tier cellular network. It is worth mentioning that if the different tiers have different path-loss exponents, the uplink association regions will follow a weighted Voronoi tessellation. However, the association in the downlink is still different from the association in the uplink due to the homogeneous transmit powers of the UEs in the latter case and the heterogeneous transmit powers of the BSs in the former case. 

In addition to its practicality, the channel inversion power control with the cutoff threshold highly facilitates the analysis in this paper and leads to simple expressions for the outage probability and rate. However, the presented analysis can be considered as a special case of the fractional channel inversion power control (i.e., full inversion) and cannot be used to capture the fractional channel inversion power control. Nevertheless, the presented analysis provides design insights and helps understanding  the uplink system behavior. In addition, in contrast to fractional channel inversion power control, the channel inversion power control with the cutoff threshold results in a location independent outage probability and average rate for active users (i.e., homogenous performance for all active users). The fractional channel inversion power control imposes a location dependent performance for the active users. Furthermore, the fractional channel inversion power control highly complicates the system model, and hence, necessitates more approximations as in \cite{and-uplink} in order to maintain the tractability.  Although \cite{and-uplink} provided a mathematically elegant technique to deal with uplink cellular networks with fractional channel inversion power control, our results reveal that the model provided in \cite{and-uplink} will not be accurate when considering the maximum transmit power constraint for the UEs and the UEs operating in a binding maximum transmit power scenario. This is because, in \cite{and-uplink} the system model was approximated and it was assumed that the Voronoi cells are realized w.r.t. the users not the BSs. In other words,  the uplink system model was converted to an equivalent downlink (i.e., the users are treated as BSs) which, as we have shown here, is inaccurate if $P_u$ is a binding operational constraint. Moreover, the model presented in \cite{and-uplink} resulted in a relatively more complex expression for the outage probability\footnote{Note that, the expressions presented in \cite{and-uplink} for the outage probability and rate do not reduce to the expression presented in this paper due to the approximation in system model used in \cite{and-uplink}.}.

\section{Conclusion}

We have introduced a novel  modeling approach for uplink transmission in Poisson cellular networks with truncated channel inversion power control. The model assumes a practical  system model and accounts for power control, maximum transmit power of the UEs, and cutoff threshold for the power control. The developed modeling approach is general and captures the uplink performance of multi-tier cellular networks. The results show that the cutoff threshold introduces tradeoffs in the system performance and there exists an optimal cutoff threshold that minimizes the outage probability. When multiple outage optimal cutoff thresholds exist, the minimum cutoff threshold minimizes the average transmit power  of the users. Approximate yet accurate simple expressions for the outage probability and spectral efficiency, which reduce to closed forms is some special cases, have been obtained. 

We have characterized the uplink performance and showed the existence of a transfer point for the uplink operation which depends on the relative values of the BS intensity, the maximum transmit power of the UEs, and the cutoff threshold for power control. When the BSs are dense enough such that the maximum transmit power is not a binding constraint to establish an uplink connection with the nearest BS, the {\rm SINR} outage and the spectral efficiency are independent of the BS intensity. Since the maximum transmit power $P_u$ is, in general, not a binding constraint in dense cellular networks, deploying more tiers and more BSs may decrease the transmit power consumption for UEs, improve the spatial frequency reuse and system capacity, but will neither improve nor deteriorate the {\rm SINR} outage probability and spectral efficiency. Hence, the {\rm SINR} outage probability and the spectral efficiency could  only be improved via interference management/avoidance techniques.

%\section*{Acknowledgment}
%This work was supported in part by an IPS from Natural Sciences and Engineering Research Council of Canada (NSERC), an NSERC Strategic  Grant (STPGP 430285), and in part by a scholarship from TR{\em Tech}, Winnipeg, Manitoba, Canada.
\appendices

\section{Proof of Lemma~\ref{lem1}} \label{proof1}

For a generic UE $u$, let $r_o=\underset{m_i \in \mathbf{\Psi}}{\min}\left\|u-m_i\right\|$. The uplink distance $r_o$ has the Rayleigh distribution $f_{r_o}(r) = 2 \pi \lambda r e^{-\pi \lambda r^2}$, $0 \leq r \leq \infty$ \cite{trac, martin-book}. The transmit power for a generic UE is given by $P = \rho_o r_o^{\eta}$ such that $0 \leq P \leq P_u$. Therefore, the {\em pdf} of $P$ is given by
 \begin{align}
f_{P}(x) &=  \frac{{2 \pi \lambda} x^{\frac{2}{\eta}-1} e^{-\pi \lambda \left(\frac{x}{\rho_o}\right)^{\frac{2}{\eta}}}}{ {\eta \rho_o^{\frac{2}{\eta}}} \int_0^{P_u} \frac{2 \pi \lambda}{\eta  \rho_o^{\frac{2}{\eta}}} y^{\frac{2}{\eta}-1} e^{-\pi \lambda  \left(\frac{y}{\rho_o}\right)^{\frac{2}{\eta}}} dy} \notag \\
&=  \frac{2 \pi \lambda x^{\frac{2}{\eta}-1} e^{-\pi \lambda \left(\frac{x}{\rho_o}\right)^{\frac{2}{\eta}}}}{\eta  \rho_o^{\frac{2}{\eta}} \left(1-e^{-\pi \lambda \left(\frac{P_u}{\rho_o}\right)^{\frac{2}{\eta}}}\right)}, 0\leq x \leq P_u. \notag 
\end{align}
Note that the {\em pdf} of $P$ is normalized due to the truncated channel inversion power control. The $\alpha^{th}$ moment of $P$ is given by $\int_0^{P_u} x^\alpha f_{P}(x) dx$ and the lemma is obtained.

\section{Proof of Theorem~\ref{th1}} \label{proof2}

For the interference experienced by a cellular uplink, we find the Laplace Transform  (LT) of the aggregate interference at a tagged BS located at the origin. Note that orthogonal channel assignment per BS brings correlations among the locations of the interfering UEs as well as with the location of the tagged BS, which highly complicates the analysis. The derivation here is based on the three facts and a key assumption listed below: 

\begin{itemize}
\item \textbf{Fact $\#1$}: the average useful signal received at any BS is equal to the cutoff threshold $\rho_o$.

\item \textbf{Fact $\#2$}: the average interference received from any single interfering UE is strictly less than $\rho_o$.

\item \textbf{Fact $\#3$}: at any time instant each BS will have a single user served per channel, and hence, the intensity of interfering UEs on each channel is $\lambda$. 

\item \textbf{Key assumption}: \textcolor{red}{the interfering UEs  constitute a PPP and their transmit powers are independent}. 
\end{itemize}

Note that \textbf{Fact $\#1$} and \textbf{Fact $\#2$} are a direct consequence of the association policy and power control, while \textbf{Fact $\#3$} holds because each BS assigns a unique channel for each of its associated users. Hence, the aggregate interference received at the tagged BS can be written as
\begin{equation} 
\mathcal{I} = \underset{u_i \in \tilde{\mathbf{\Phi}}\setminus \{o\}}{\sum} \mathbbm{1}\left(P_i \left\|u_i\right\|^{-\eta} < \rho_o\right) P_i h_i \left\|u_i\right\|^{-\eta}
\label{corri}
\end{equation}
where $\tilde{\mathbf{\Psi}}$ is a PPP with intensity $ \lambda$ representing the interfering UEs, and $\mathbbm{1}(.)$ is the indicator function which takes the value $1$ if the statement $(.)$ is true and zero otherwise. Note that the indicator function in (\ref{corri}) captures the correlation between the interfering UEs and the tagged BS. The LT of the aggregate interference from the interfering UEs received at the tagged BS is obtained as
\small
\begin{align}
  & \mathcal{L}_{\mathcal{I}}(s) = \mathbb{E}\left[e^{-  s \mathcal{I}}\right]\notag \\
	&= \mathbb{E}\left[e^{-  s\underset{u_i \in \tilde{\mathbf{\Phi}}\setminus \{o\}}{\sum} \mathbbm{1}\left(P_i \left\|u_i\right\|^{-\eta} < \rho_o\right) P_i h_i \left\|u_i\right\|^{-\eta}}\right]\notag \\
&\stackrel{(i)}{=}  \mathbb{E}_{\tilde{\mathbf{\Phi}}}\left[\underset{u_i \in \tilde{\mathbf{\Phi}}\setminus \{o\}}{\prod} \mathbb{E}_{P_i,h_i}\left[e^{-s \mathbbm{1}\left( \left\|u_i\right\|> \left(\frac{P_i}{\rho_o}\right)^\frac{1}{\eta} \right) P_i h_i \left\|u_i\right\|^{-\eta}}\right]\right]\notag \\
 &\stackrel{(ii)}{=}  \exp\left( - 2 \pi   \lambda \int_{\left(\frac{P}{\rho_o}\right)^\frac{1}{\eta}}^{\infty}\mathbb{E}_{P,h}\left[  \left(1- e^{-s P h x^{-\eta}}\right)\right] xdx \right)\notag \\
 & \stackrel{(iii)}{=}  \exp\left( - 2 \pi    \lambda\int_{\left(\frac{P}{\rho_o}\right)^\frac{1}{\eta}}^{\infty}\mathbb{E}_{P}\left[  \left(1- \frac{1}{1+s P x^{-\eta}}\right)\right] xdx \right)\notag \\
 %& = \exp\left( - 2 \pi    \lambda\int_{\left(\frac{P}{\rho_o}\right)^\frac{1}{\eta}}^{\infty}\mathbb{E}_{P}\left[ \frac{s P x}{x^{\eta}+s P }\right]dx \right)\notag \\
&\stackrel{(iv)}{=} \exp\left( - {2 \pi    \lambda} {s^\frac{2}{\eta}} \mathbb{E}_{P}\left[ P^\frac{2}{\eta}\right] \int_{(s \rho_o)^\frac{-1}{\eta}}^{\infty} \frac{y}{y^{\eta}+1 }dy \right) 
\label{lt1}
\end{align}
\normalsize
where $\mathbb{E}_{x}[.]$ is the expectation with respect to the random variable $x$, $(i)$ follows from independence between $\tilde{\mathbf{\Phi}}$, $P_i$, and $h_i$, $(ii)$ follows from the probability generation functional (PGFL) of the PPP \cite{martin-book}, $(iii)$ follows from the LT of $h$, and $(iv)$ obtained by changing the variables $y=\frac{x}{(SP)\frac{1}{\eta}}$. The theorem is obtained be substituting (\ref{lt1}) in (\ref{SINR1}) for $s=\frac{\theta}{\rho_o}$ and substituting the value of $\mathbb{E}[P^{\frac{2}{\eta}}]$ from \textbf{Lemma~\ref{lem1}}. It is worth mentioning that $(i)$ follows from the assumption of i.i.d. transmit power, and $(ii)$ follows from the PPP assumption for the UEs' locations (i.e., using the PGFL of the PPP).

\section{Proof of Theorem~\ref{th2}} \label{proof3}

Since the {\rm SINR} is a strictly positive random variable, the average spectral efficiency can be obtained as

\begin{align}
\mathcal{R} &=\mathbb{E}[\ln\left(1+{\rm SINR}\right)] \notag \\
 &= \int_0^{\infty} \mathbb{P}\left\{\ln\left(1+{\rm SINR}\right) > t\right\} dt \notag \\
%&= \int_0^{\infty} \mathbb{P}\left\{{\rm SINR} > \left(e^t-1\right)\right\} dt \notag \\
&\stackrel{(v)}{=} \int_0^{\infty} e^{-\frac{\left(e^t-1\right) \sigma^2}{\rho_o} } \mathcal{L}_{\mathcal{I}}\left(\frac{ \left(e^t-1\right)}{\rho_o}\right) dt \notag \\
&\stackrel{(vi)}{=} \int_0^{\infty} \frac{1}{x+1}e^{-\frac{x \sigma^2}{\rho_o}} \mathcal{L}_{\mathcal{I}}\left(\frac{x}{\rho_0}\right) dx 
\label{rate}
\end{align}  
where $(v)$ follows from (\ref{SINR}), and $(vi)$ is obtained by changing the variables $x=\left(e^t-1\right)$. The theorem is obtained by substituting the value of $\mathcal{L}_{\mathcal{I}}\left(s\right) $ from (\ref{lt1}) in (\ref{rate}) for $s=\frac{x}{\rho_o}$ and substituting the value of $\mathbb{E}[P^{\frac{2}{\eta}}]$ from \textbf{Lemma~\ref{lem1}}.

\section{Proof of Theorem~\ref{th3}} \label{proofth3}

For the interference experienced by a cellular uplink from UEs in tier $j$, we find the LT of the aggregate interference at a tagged BS located at the origin. Similar to \textbf{Appendix~\ref{proof2}}, this proof is based on the 3 facts and a key assumption listed below:

\begin{itemize}
\item \textbf{Fact $\#1$}: the average useful signal received at any BS in the $j^{th}$ tier is equal to the cutoff threshold $\rho_o^{(j)}$.

\item \textbf{Fact $\#2$}: the average interference received from any single interfering UE in the $k^{th}$ tier is strictly less than $\rho_o^{(k)}$.

\item \textbf{Fact $\#3$}: at any time instant each BS in tier $k$ will have a single user served per channel, and hence, the intensity of interfering UEs from the $k^{th}$ tier on each channel is $\lambda_k$. 

\item \textbf{Key assumption}: \textcolor{red}{the interfering UEs from the $k^{th}$ tier  constitute a PPP with intensity $\lambda_k$ and their transmit powers are independent}. 
\end{itemize}

Hence, the aggregate interference from UEs in tier $k$ received at the tagged BS can be written as
\begin{equation} 
\mathcal{I}_k = \underset{u_i \in \tilde{\mathbf{\Phi}_k}\setminus \{o\}}{\sum} \mathbbm{1}\left(P_{ki} \left\|u_i\right\|^{-\eta} < \rho^{(k)}_o\right) P_{ki} h_i \left\|u_i\right\|^{-\eta}
\label{corri2}
\end{equation}
where $\tilde{\mathbf{\Psi}}_k$ is a PPP with intensity $\lambda_k $ representing the interfering UEs. Note that the indicator function in (\ref{corri2}) captures the correlation among the locations of the interfering UEs and the location of the tagged BS. The LT of the aggregate interference from UEs in tier $k$ received at the tagged BS in tier $j$ is obtained as
\small
\begin{align}
 &\mathcal{L}_{\mathcal{I}_k}(s) = \mathbb{E}\left[e^{-  s\underset{u_i \in \tilde{\mathbf{\Phi}_k}\setminus \{o\}}{\sum} \mathbbm{1}\left(P_{ki} \left\|u_i\right\|^{-\eta} < \rho^{(k)}_o\right) P_{ki} h_i \left\|u_i\right\|^{-\eta}}\right]\notag \\
%&= \mathbb{E}_{\tilde{\mathbf{\Phi}}_k}\left[\underset{u_i \in \tilde{\mathbf{\Phi}}_k\setminus \{o\}}{\prod} \mathbb{E}_{P_{ki},h_i}\left[e^{-s \mathbbm{1}\left( \left\|u_i\right\|> \left(\frac{P_{ki}}{\rho^{(k)}_o}\right)^\frac{1}{\eta} \right) P_{ki} h_i \left\|u_i\right\|^{-\eta}}\right]\right]\notag \\
 & = \exp\left( - 2 \pi  _k  \lambda_k \int_{\left(\frac{P_k}{\rho^{(k)}_o}\right)^\frac{1}{\eta}}^{\infty}\mathbb{E}_{P_k,h}\left[  \left(1- e^{-s P_k h x^{-\eta}}\right)\right] xdx \right)\notag \\
% & = \exp\left( - 2 \pi  _k  \lambda_k \int_{\left(\frac{P_k}{\rho^{(k)}_o}\right)^\frac{1}{\eta}}^{\infty}\mathbb{E}_{P_k}\left[  \left(1- \frac{1}{1+s P_k x^{-\eta}}\right)\right] xdx \right)\notag \\
 & = \exp\left( - 2 \pi  _k  \lambda_k\int_{\left(\frac{P_k}{\rho^{(k)}_o}\right)^\frac{1}{\eta}}^{\infty}\mathbb{E}_{P_k}\left[ \frac{s P_k x}{x^{\eta}+s P_k }\right]dx \right)\notag \\
&= \exp\left( - {2 \pi  _k  \lambda_k} {s^\frac{2}{\eta}} \mathbb{E}_{P_k}\left[ P_k^\frac{2}{\eta}\right] \int_{(s \rho^{(k)}_o)^\frac{-1}{\eta}}^{\infty} \frac{y}{y^{\eta}+1 }dy \right). 
\label{lt2}
\end{align}
\normalsize
The theorem is obtained by substituting (\ref{lt2}) in (\ref{SINR2}) for $s=\frac{\theta}{\rho^{(j)}_o}$ and substituting the value of $\mathbb{E}[P_k^{\frac{2}{\eta}}]$ from \textbf{Lemma~\ref{lem2}}.

\section{Proof of Lemma~\ref{lem5}} \label{lemm5}

Let $r_k = \underset{m^{(k)}_i \in \mathbf{\Psi}_k}{\min}\left(\left\|u-m^{(k)}_i\right\|\right)$ be the distance from a UE to its nearest BSs from each of the coexisting tiers. Then, $f_{r_k}(r)= 2 \pi \lambda_k r e^{-\pi \lambda_k r^2}$, $0 \leq r \leq \infty$. Since a UE connects to the BS with the best link quality, then given that a generic UE $u$ is connected to a generic BS from the $j^{th}$ tier, we have $r_j^{\eta_j} \leq r_k^{\eta_k}$ $\forall k$, $k \neq j$. Using this fact, we can write $r_j^{\eta_j} =  \underset{k}{\min}(r_k^{\eta_k})$, and hence, the transmit power of a generic UE connected to a generic BS in the $j^{th}$ tier is given by $P_j= \rho^{(j)}_o  \underset{k}{\min}(r_k^{\eta_k})$ such that $0 \leq P_j \leq P_u$. Hence, the {\em cdf} of the transmit power can be written as  

%\begin{align}
%F_{p_k}(x) &= 1-\prod_{k=1}^{k}\left(\frac{e^{-\pi \lambda \left(\frac{x}{\rho_o}\right)^{\frac{2}{\eta_k}}}}{1-e^{-\pi \lambda \left(\frac{P_u}{\rho_o}\right)^{\frac{2}{\eta_k}}}}\right)
%\end{align}

\begin{align}
F_{P_j}(x) &=  \frac{1- e^{-\sum_{k=1}^{K}\pi \lambda_k \left(\frac{x}{\rho_o^{(j)}}\right)^{\frac{2}{\eta_k}}}}{1-e^{-\sum_{b=1}^{K}\pi \lambda_b \left(\frac{P_u}{\rho^{(j)}_o}\right)^{\frac{2}{\eta_b}}}}
\end{align}
and the {\em pdf} of the transmit power is given by
\begin{align}
f_{P_j}(x) &=  \frac{\sum_{k=1}^{K} \frac{2 \pi \lambda_k x^{\frac{2}{\eta_k}-1} }{\eta_k (\rho_o^{(j)})^\frac{2}{\eta_k} }}{1-e^{-\sum_{b=1}^{K}\pi \lambda_b \left(\frac{P_u}{\rho^{(j)}_o}\right)^{\frac{2}{\eta_b}}}} {e^{-\sum_{a=1}^{K}\pi \lambda_a \left(\frac{x}{\rho_o^{(j)}}\right)^{\frac{2}{\eta_a}}}}.
\end{align}

The $\alpha^{th}$ moment of $P_k$ cannot be obtained in closed from except for a common path-loss exponent $\eta$.

\section{Proof of Theorem~\ref{th5}} \label{proof5}

This proof is based on the 3 facts and a key assumption listed in \textbf{Appendix~\ref{proofth3}}. The aggregate interference from UEs in tier $k$ received at the tagged BS from tier $j$ can be written as
\begin{equation} 
\mathcal{I}_k = \underset{u_i \in \tilde{\mathbf{\Phi}_k}\setminus \{o\}}{\sum} \mathbbm{1}\left(P_{ki} \left\|u_i\right\|^{-\eta_j} < \rho^{(k)}_o\right) P_{ki} h_i \left\|u_i\right\|^{-\eta_j}.
\label{diff}
\end{equation}
Note that, although the path-loss exponents are different, as long as the UEs associate based on the link quality, (\ref{diff}) will hold ture. Similar to \textbf{Appendix~\ref{proofth3}}, after some mathematical manipulations, the theorem can be proved. %

\end{document}